\documentclass[11pt,letterpaper]{article}
\usepackage[utf8x]{inputenc}
\usepackage{quoting}
\quotingsetup{vskip=0pt,leftmargin=20pt,rightmargin=20pt}

\def\showauthornotes{0}
\def\showtableofcontents{0} 
\def\showkeys{0}
\def\showdraftbox{0}

\def\usemicrotype{1}
\def\showfixme{0}

\usepackage{etex}

\usepackage[l2tabu, orthodox]{nag}

\usepackage{xspace,enumerate}

\usepackage[dvipsnames]{xcolor}

\usepackage[T1]{fontenc}
\usepackage[full]{textcomp}

\usepackage[american]{babel}

\usepackage{mathtools}

\usepackage{amsthm}

\newtheorem{theorem}{Theorem}[section]
\newtheorem*{theorem*}{Theorem}

\newtheorem*{proposition*}{Proposition}
\newtheorem{lemma}[theorem]{Lemma}
\newtheorem*{lemma*}{Lemma}
\newtheorem{corollary}[theorem]{Corollary}
\newtheorem*{conjecture*}{Conjecture}
\newtheorem{fact}[theorem]{Fact}
\newtheorem*{fact*}{Fact}

\newtheorem*{hypothesis*}{Hypothesis}

\theoremstyle{definition}
\newtheorem{definition}[theorem]{Definition}

\theoremstyle{remark}
\newtheorem{claim}[theorem]{Claim}
\newtheorem*{claim*}{Claim}

\newtheorem*{remark*}{Remark}

\newtheorem*{observation*}{Observation}

\usepackage[
letterpaper,
top=0.7in,
bottom=0.9in,
left=1in,
right=1in]{geometry}

\usepackage[varg]{pxfonts}
\usepackage{textcomp}
\usepackage[scr=rsfso]{mathalfa}
\usepackage{bm} 
\let\mathbb\varmathbb

\ifnum\showkeys=1
\usepackage[color]{showkeys}
\fi

\usepackage{hyperref}
\hypersetup{
pagebackref,
colorlinks=true,
urlcolor=blue,
linkcolor=blue,
citecolor=OliveGreen,
}
\usepackage{prettyref}

\newcommand{\savehyperref}[2]{\texorpdfstring{\hyperref[#1]{#2}}{#2}}

\newrefformat{eq}{\savehyperref{#1}{\textup{(\ref*{#1})}}}
\newrefformat{lem}{\savehyperref{#1}{Lemma~\ref*{#1}}}
\newrefformat{def}{\savehyperref{#1}{Definition~\ref*{#1}}}
\newrefformat{thm}{\savehyperref{#1}{Theorem~\ref*{#1}}}
\newrefformat{cor}{\savehyperref{#1}{Corollary~\ref*{#1}}}
\newrefformat{cha}{\savehyperref{#1}{Chapter~\ref*{#1}}}
\newrefformat{sec}{\savehyperref{#1}{Section~\ref*{#1}}}
\newrefformat{app}{\savehyperref{#1}{Appendix~\ref*{#1}}}
\newrefformat{tab}{\savehyperref{#1}{Table~\ref*{#1}}}
\newrefformat{fig}{\savehyperref{#1}{Figure~\ref*{#1}}}
\newrefformat{hyp}{\savehyperref{#1}{Hypothesis~\ref*{#1}}}
\newrefformat{alg}{\savehyperref{#1}{Algorithm~\ref*{#1}}}
\newrefformat{rem}{\savehyperref{#1}{Remark~\ref*{#1}}}
\newrefformat{item}{\savehyperref{#1}{Item~\ref*{#1}}}
\newrefformat{step}{\savehyperref{#1}{step~\ref*{#1}}}
\newrefformat{conj}{\savehyperref{#1}{Conjecture~\ref*{#1}}}
\newrefformat{fact}{\savehyperref{#1}{Fact~\ref*{#1}}}
\newrefformat{prop}{\savehyperref{#1}{Proposition~\ref*{#1}}}
\newrefformat{prob}{\savehyperref{#1}{Problem~\ref*{#1}}}
\newrefformat{claim}{\savehyperref{#1}{Claim~\ref*{#1}}}
\newrefformat{relax}{\savehyperref{#1}{Relaxation~\ref*{#1}}}
\newrefformat{red}{\savehyperref{#1}{Reduction~\ref*{#1}}}
\newrefformat{part}{\savehyperref{#1}{Part~\ref*{#1}}}

\newcommand{\Sref}[1]{\hyperref[#1]{\S\ref*{#1}}}

\usepackage{nicefrac}

\let\nfrac=\nicefrac

\ifnum\usemicrotype=1
\usepackage{microtype}
\fi

\ifnum\showauthornotes=1
\newcommand{\Authornote}[2]{{\sffamily\small\color{red}{[#1: #2]}}}
\newcommand{\Authornotecolored}[3]{{\sffamily\small\color{#1}{[#2: #3]}}}
\newcommand{\Authorcomment}[2]{{\sffamily\small\color{gray}{[#1: #2]}}}
\newcommand{\Authorstartcomment}[1]{\sffamily\small\color{gray}[#1: }

\newcommand{\Authorfnote}[2]{\footnote{\color{red}{#1: #2}}}
\newcommand{\Authorfixme}[1]{\Authornote{#1}{\textbf{??}}}
\newcommand{\Authormarginmark}[1]{\marginpar{\textcolor{red}{\fbox{\Large #1:!}}}}
\else
\newcommand{\Authornote}[2]{}
\newcommand{\Authornotecolored}[3]{}
\newcommand{\Authorcomment}[2]{}
\newcommand{\Authorstartcomment}[1]{}

\newcommand{\Authorfnote}[2]{}
\newcommand{\Authorfixme}[1]{}
\newcommand{\Authormarginmark}[1]{}
\fi

\newcommand{\PRnote}{\Authornote{PR}}

\newcommand{\RMnote}{\Authornote{RM}}

\ifnum\showfixme=0

\fi

\usepackage{boxedminipage}

\newcommand{\norm}[1]{\lVert#1\rVert}

\newcommand{\iprod}[1]{\langle#1\rangle}

\newcommand{\Esymb}{\mathbb{E}}
\newcommand{\Psymb}{\mathbb{P}}

\DeclareMathOperator*{\E}{\Esymb}

\DeclareMathOperator*{\ProbOp}{\Psymb}

\renewcommand{\Pr}{\ProbOp}

\newcommand{\defeq}{\stackrel{\mathrm{def}}=}

\newcommand{\seteq}{\mathrel{\mathop:}=}

\newcommand{\mper}{\,.}
\newcommand{\mcom}{\,,}

\newcommand\bdot\bullet

\newcommand*{\transposed}{\mkern-1mu\intercal}

\ifx\mathds\undefined 
\newcommand{\Ind}{\mathbb I}
\else
\newcommand{\Ind}{\mathds 1}
\fi

\DeclareMathOperator{\opt}{opt}

\DeclareMathOperator{\rank}{rank}

\newcommand{\etal}{et al.\xspace}

\newcommand{\N}{\mathbb N}
\newcommand{\R}{\mathbb R}

\newcommand{\cA}{\mathcal A}
\newcommand{\cB}{\mathcal B}
\newcommand{\cC}{\mathcal C}

\newcommand{\cI}{\mathcal I}

\newcommand{\cL}{\mathcal L}
\newcommand{\cM}{\mathcal M}

\newcommand{\cR}{\mathcal R}

\renewcommand{\leq}{\leqslant}

\renewcommand{\geq}{\geqslant}

\ifnum\showdraftbox=1
\newcommand{\draftbox}{\begin{center}
  \fbox{%
    \begin{minipage}{2in}%
      \begin{center}%
          \Large\textsc{Working Draft}\\%
        Please do not distribute%
      \end{center}%
    \end{minipage}%
  }%
\end{center}
\vspace{0.2cm}}
\else
\newcommand{\draftbox}{}
\fi

\let\epsilon=\varepsilon

\numberwithin{equation}{section}

\newcommand\MYcurrentlabel{xxx}

\newcommand{\MYstore}[2]{%
  \global\expandafter \def \csname MYMEMORY #1 \endcsname{#2}%
}

\newcommand{\MYload}[1]{%
  \csname MYMEMORY #1 \endcsname%
}

\newcommand{\MYnewlabel}[1]{%
  \renewcommand\MYcurrentlabel{#1}%
  \MYoldlabel{#1}%
}

\newcommand{\MYdummylabel}[1]{}

\newcommand{\torestate}[1]{%
  \let\MYoldlabel\label%
  \let\label\MYnewlabel%
  #1%
  \MYstore{\MYcurrentlabel}{#1}%
  \let\label\MYoldlabel%
}

\newcommand{\restatetheorem}[1]{%
  \let\MYoldlabel\label
  \let\label\MYdummylabel
  \begin{theorem*}[Restatement of \prettyref{#1}]
    \MYload{#1}
  \end{theorem*}
  \let\label\MYoldlabel
}

\newcommand{\restatelemma}[1]{%
  \let\MYoldlabel\label
  \let\label\MYdummylabel
  \begin{lemma*}[Restatement of \prettyref{#1}]
    \MYload{#1}
  \end{lemma*}
  \let\label\MYoldlabel
}

\newcommand{\restateprop}[1]{%
  \let\MYoldlabel\label
  \let\label\MYdummylabel
  \begin{proposition*}[Restatement of \prettyref{#1}]
    \MYload{#1}
  \end{proposition*}
  \let\label\MYoldlabel
}

\newcommand{\restatefact}[1]{%
  \let\MYoldlabel\label
  \let\label\MYdummylabel
  \begin{fact*}[Restatement of \prettyref{#1}]
    \MYload{#1}
  \end{fact*}
  \let\label\MYoldlabel
}

\newcommand{\restate}[1]{%
  \let\MYoldlabel\label
  \let\label\MYdummylabel
  \MYload{#1}
  \let\label\MYoldlabel
}

\newcommand{\addreferencesection}{
  \phantomsection
  \addcontentsline{toc}{section}{References}
}

\newcommand{\eps}{\epsilon}

\let\origparagraph\paragraph
\renewcommand{\paragraph}[1]{\origparagraph{#1.}}
\allowdisplaybreaks
\usepackage{paralist}

\usepackage{comment}

\let\citet\cite
\theoremstyle{definition}

\usepackage{relsize}
\usepackage[font=footnotesize]{caption}
\usepackage{yfonts}

\newcommand{\pE}{\tilde {\mathbb E}}

\let\cL\relax
\DeclareMathOperator{\cL}{\mathcal L}

\DeclareUrlCommand\email{}

\newcommand{\iproduct}{g}

\newcommand{\pnote}[1]{\textcolor{blue}{\small{[Pravesh: #1]}}}
\usepackage{appendix}
\usepackage{amssymb}
\usepackage{amsfonts}
\usepackage{latexsym}
\usepackage{amsmath}

\usepackage[T1]{fontenc}
\usepackage{fullpage,appendix}
\usepackage{dsfont}
\usepackage[capitalize]{cleveref}

\usepackage{color}
\usepackage{tikz}

\newcommand{\ignore}[1]{}
\definecolor{corlinks}{RGB}{64,128,128}
\definecolor{cormenu}{RGB}{0,37,94}
\definecolor{corurl}{RGB}{0,46,91}

\newcommand{\on}{\{-1,1\}}
\newcommand{\1}{\mathds{1}}

\newcommand{\bkets}[1]{\left(#1\right)}
\newcommand{\sbkets}[1]{\left[#1\right]}

\renewcommand{\P}{\mathcal{P}}

\DeclareMathOperator*{\ex}{\mathbb{E}}
\DeclareMathOperator*{\pr}{\mathsf{Pr}}
\renewcommand{\L}{\mathcal{L}}

\newcommand{\I}{\mathcal{I}}

\renewcommand{\S}{\mathbb{S}}

\newcommand{\zo}{\{0, 1\}}

\renewcommand{\cal}[1]{\mathcal{#1}}
\renewcommand{\int}{\mathsf{int}}

\renewcommand{\1}{\mathbf{1}}

\renewcommand{\hat}{\widehat}

\renewcommand{\on}[1]{\tfrac{\binom{\omega}{#1}}{\binom{n}{#1}}}

\renewcommand{\emptyset}{\varnothing}

\makeatletter
\let\orgdescriptionlabel\descriptionlabel
\renewcommand*{\descriptionlabel}[1]{%
  \let\orglabel\label
  \let\label\@gobble
  \phantomsection
  \edef\@currentlabel{#1}%
  \let\label\orglabel
  \orgdescriptionlabel{#1}%
}
\makeatother

\newcommand{\gtn}{g^{\otimes n}}

\newcommand{\mnote}[1]{}
\newcommand{\degp}{\mathsf{deg}_{+}}
 \newcommand{\newref}[2][]{\hyperref[#2]{#1~\ref*{#2}}}

\newcommand{\lref}[1]{\newref[Lemma]{#1}}

\newcommand{\eref}[1]{\newref[Equation]{#1}}

\renewcommand{\eqref}[1]{\hyperref[#1]{(\ref*{#1})}}

\usepackage{latexsym}

\usepackage{dsfont}
\usepackage{color}
\usepackage{tikz}

\definecolor{corlinks}{RGB}{64,128,128}
\definecolor{cormenu}{RGB}{0,37,94}
\definecolor{corurl}{RGB}{0,46,91}

\newcommand{\bd}{\textsc{blockwise-dense }}
\newcommand{\cbd}{\textsc{CBD} }

\newcommand{\lfta}{\leftarrow}
\renewcommand{\cal}[1]{\mathcal{#1}}

\def\inpw#1,#2{\langle #1, #2\rangle}

\newcommand{\nnr}{\mathsf{nnr}}

\renewcommand{\nnr}{\mathsf{rank}_+}
\newcommand{\SA}{\mathsf{SA}}
\newcommand*{\trsp}{\mkern-1mu\intercal}

\newcommand{\Dec}{\textsc{Decompose }}
\newcommand{\YDec}{\textsc{YDecompose }}
\newcommand{\XDec}{\textsc{XDecompose }}

\newcommand{\Errorb}{\mathsf{Error}_b}
\newcommand{\Errora}{\mathsf{Error}_a}
\newcommand{\Error}{\mathsf{Error}}
\renewcommand{\bar}[1]{\overline{#1}}

\renewcommand{\on}{\{-1,1\}}
\renewcommand{\deg}{\mathsf{deg}}
\title{Approximating Rectangles by Juntas and Weakly-Exponential Lower Bounds for LP Relaxations of CSPs}

\author{
Pravesh K. Kothari
\thanks{\url{kothari@cs.princeton.edu}. Part of the work was done while the author visited UC Berkeley.}
\\ IAS and Princeton University, NJ, USA
\and Raghu Meka
\thanks{\url{raghum@cs.ucla.edu}. Supported by NSF CCF-1553605. Part of this talk was done while the author was visiting the Simons Institute for Theory of Computing, Berkeley as part of the program on Fine-grained Complexity.}
\\University of California, Los Angeles, CA, USA
\and Prasad Raghavendra
\thanks{\url{prasad@cs.berkeley.edu}. Supported by NSF CCF-1408643, CCF-1343104 and the Sloan Fellowship}
\\University of California, Berkeley, CA, USA}

\date{}

\begin{document}

\maketitle
\draftbox
\thispagestyle{empty}

\begin{abstract}
We show that for constraint satisfaction problems (CSPs), weakly-exponential size linear programming relaxations are as powerful as $n^{\Omega(1)}$-rounds of the Sherali-Adams linear programming hierarchy. As a corollary, we obtain sub-exponential size lower bounds for linear programming relaxations that beat random guessing for many CSPs such as MAX-CUT and MAX-3SAT. This is a nearly-exponential improvement over previous results; previously, it was only known that linear programs of size $n^{o(\log n)}$ cannot beat random guessing for any CSP \cite{CLRS13}.

Our bounds are obtained by exploiting and extending the recent progress in communication complexity for "lifting" query lower bounds to communication problems. The main ingredient in our results is a new structural result on ``high-entropy rectangles'' that may of independent interest in communication complexity.

\end{abstract}
\clearpage

\ifnum\showtableofcontents=1
{
\tableofcontents
\thispagestyle{empty}
 }
\fi

\clearpage

\setcounter{page}{1}
\newcommand{\rnng}{\R_{\geq 0}}

\section{Introduction}
Translating a combinatorial problem over a discrete domain to a problem in continuous space has been an important concept in computer science over the last few decades; in this vein, linear programming relaxation is one of the most used techniques for algorithm design. In this work we prove limitations on the power of linear programs (LPs) as applied to \emph{constraint satisfaction problems} (CSPs).

Constraint satisfaction problems such as MAX-3SAT or MAX-3XOR or MAX-CUT are some of the most well-studied problems in approximation algorithms as well as combinatorial optimization. Here we show unconditional lower bounds for approximately solving CSPs by LPs. Informally, we show that for many CSPs such as MAX-3SAT, MAX-3XOR, or MAX-CUT, no LP of size $2^{n^{\Omega(1)}}$ can beat the \emph{trivial} approximation factor ($7/8$ for MAX-3SAT, $1/2$ for MAX-3XOR, $1/2$ for MAX-CUT); we also show similar results for vertex-cover. Previously, such lower bounds only applied to LPs of size at most $n^{\Omega((\log n)/(\log \log n))}$ \cite{CLRS13}.

The core of our result above is a new structural result about rectangles that has various applications in communication complexity in the context of \emph{lifting query} lower bounds to communication lower bounds.

\subsection{CSPs, Linear programming relaxations, Sherali-Adams hierarchy} \label{sec:lpforCSP}
A MAX-CSP (henceforth referred to only as CSP) is defined by a \emph{predicate} $P: \on^k \rightarrow \{0,1\}$. An instance of the CSP, $\cal{I}$, is defined by a collection of $k$-tuples of literals $C_1, C_2, \ldots, C_m$ on $n$ Boolean variables $(x_1, x_2, \ldots, x_n)$ \footnote{Throughout this article, we will use $\{-1,1\}$ to denote Boolean inputs}. The algorithmic problem is to find an assignment to the variables $x = (x_1,\ldots,x_n)$ so as to maximize the number of satisfied constraints:

\begin{equation}\label{eq:cspdef}
\opt(\cal{I}) = \max_{x \in \on^n} \sum_{i=1}^m P(C_i(x)) \equiv \max_{x \in \on^n} \cal{I}(x),
\end{equation}
where we define $\cal{I}(x)  = \sum_{i = 1}^m P(C_i(x))$.

For example, MAX-CUT corresponds to the case where the predicate $P:\on^2 \to \zo$ is defined by $P(a,b) = (1 - ab)/2$ with instances corresponding to graphs.

Here we consider a broad-class of linear programming relaxations for CSPs obtained by \emph{linearizing} the objective function $\cal{I}(x)$. Formally, given a predicate $P$, and an integer $D$, we want:

\begin{definition}[Linearization of a CSP]
\begin{enumerate}
\item A vector $v_x \in \R^D$ for every $x \in \on^n$.
\item A vector $w_\cal{I} \in \R^D$ for every instance $\cal{I}$ of the CSP.
\item For every assignment $x$ and every instance $\cal{I}$, $\cal{I}(x) = \iprod{w_\cal{I},{v_x}}$.
\end{enumerate}
\end{definition}

Given a linearization as above, we can define a relaxation of the CSP as follows.  For a polytope $\cal{P} \subseteq \R^D$ with $\{v_x: x \in \on^n \} \subseteq \cal{P}$, we look at the linear program
$$\opt_{\cal{P}}(\cal{I}) = \max_{y \in \cal{P}} \iprod{w_\cal{I},{y}}.$$
Clearly, $\opt(\cal{I}) \leq \opt_{\cal{P}}(\cal{I})$. The \emph{complexity} or \emph{size} of the relaxation is defined as the number of facets (or inequalities) needed to describe the polytope $\cal{P}$.

\textbf{Approximating CSPs by LP relaxations.} Consider a CSP defined by a predicate $P:\on^k \to \zo$. A LP relaxation for the CSP is a sequence of polytopes $\cal{P} \equiv \{\cal{P}_n: n \geq 1\}$ where for each $n \geq 1$, $\cal{P}_n$ is a relaxation for $n$-variable instances of the CSP as defined above. For a function $s:\N \to \N$, we say $\cal{P}$ has size at most $s(n)$ if each $\cal{P}_n$ has size at most $s(n)$. 

For $0 < c \leq s \leq 1$, we say $\cal{P}$ achieves a $(c,s)$-approximation for the CSP if for $n$-variable instances $\I$ with $\opt(\I) \leq s$, $\opt_{\cal{P}_n}(\I) \leq c$. Similarly, for $0 \leq \alpha \leq 1$, we say $\cal{P}$ achieves a $\alpha$-approximation if for all $n \geq 1$, $\opt(\I) \geq \alpha \cdot \opt_{\cal{P}_n}(\I)$. In the latter case, we also say $\cal{P}$ has \emph{integrality-gap} at most $(1/\alpha)$. 

\PRnote{ToDo: define relaxation as a sequence of polytopes?}

The above framework introduced in the work of \cite{CLRS13} generalizes the \emph{extended formulation} framework of Yannakakis \cite{Yan88} and its adaptation to approximation algorithms as formulated in \cite{BFPS15}.  Furthermore, LPs arising out of the Lovasz-Schriver (LS) \cite{LovaszS} or the Sherali-Adams \cite{SheraliA} hierarchies are captured within this framework.

We prove that despite their apparent generality, when it comes to CSPs, general linear programs as above, and hence all extended formulations, are only as powerful as those obtained from the Sherali-Adams hierarchy:

\begin{theorem}\label{thm:maincsp}
There exist constants $0 < h < H$ such that the following holds. Consider a function $f:\mathbb{N} \to \mathbb{N}$. Suppose that the $f(n)$-round Sherali-Adams relaxation for a CSP cannot achieve a $(c,s)$-approximation on instances on $n$ variables. Then, no LP relaxation of size at most $n^{hf(n)}$ can achieve a $(c,s)$-approximation for the CSP on $n^H$ variables.
\end{theorem}

Charikar, Makarychev, and Makarychev \cite{CharikarMM09} showed that for all $\epsilon > 0$, there is a constant $\gamma(\epsilon)$ such that $n^{\gamma(\epsilon)}$-round Sherali-Adams relaxation for MAX-CUT has integrality gap at least $2-\epsilon$. Similarly, it follows from the works of Grigoriev \cite{Gri01} (and from that of Schoenebeck \cite{Schoenebeck08}) that $\Omega_\epsilon(n)$-round Sherali-Adams relaxations have integrality gap at least $2-\epsilon$, $8/7 - \epsilon$ for MAX-3XOR and MAX-3SAT respectively. As a corollary, we get the following lower bounds for solving CSPs by linear programming relaxations. 

\RMnote{We don't need to have precise theorem statements of the SA-lower bounds in the intro which is getting too long. Also, better to move the statement for pairwise-independence to the main body.}
\ignore{
Charikar et. al. \cite{CharikarMM09} showed the following lower bound for Max Cut and Max 3SAT problems. 

\begin{theorem}[Sherali Adams Integrality Gaps \cite{CharikarMM09}]
\label{thm:sherali-adams-cmm}
For every $\epsilon > 0$, there is a $\gamma = \gamma(\epsilon)$ such that the integrality gap of $n^{\gamma}$-degree Sherali-Adams relaxation for the Max Cut problem is at least $2-\epsilon$. 
\end{theorem}

Grigoriev \cite{Gri01} showed a strong lower bound for the Sum-of-Squares SDP hierarchy (that is a strengthening of the Sherali-Adams LP hierarchy and thus the lower bounds carry over) for 3XOR and Schoenebeck rediscovered this result and also observed that it implies a similar lower bound for the Max 3SAT problem. Benabbas et. al. \cite{BGMT12} extended this result to show a $\Omega(n)$-degree lower bound for every pairwise independent constraint satisfaction problem (i.e. a Max CSP with a predicate $P:\zo^k \rightarrow \zo$ such that there's a balanced pairwise independent distribution $\mu$ supported on $P^{-1}(1)$). 

\begin{theorem}[\cite{Gri01,Scho08,BGMT12}]
For every $k$-ary pairwise independent predicate $P$ and $\epsilon > 0$,  there exists a constant $c = c(k, \epsilon)$ such that the $cn$-degree Sherali-Adams relaxation for MAX-CSP problem on predicate $P$ has an integrality gap of at least $2^k/|P^{-1}(1)| - \epsilon$. As a corollary, $\Omega_{\epsilon}(n)$-degree Sherali-Adams relaxation of Max 3SAT has an integrality gap of at least $8/7-\epsilon$.
\end{theorem}}

\begin{corollary}\label{cor:maincsp}
For some universal constant $H \geq 1$, for every $\epsilon > 0$, there exist constants $c_1(\epsilon), c_2(\epsilon), c_3(\epsilon)$ such that the following hold: no LP relaxation of size less than $2^{c_1(\epsilon) n^{1/H}}$ has integrality gap less than $(8/7-\epsilon)$ for MAX-3SAT; no LP relaxation of size less than $2^{c_2(\epsilon) n^{1/H}}$ has integrality gap less than $(2-\epsilon)$ for MAX-3XOR; no LP relaxation of size less than $2^{n^{c_3(\epsilon)}}$ has integrality gap less than $(2-\epsilon)$ for Max-CUT. 
\end{corollary}

We also get similar bounds more generally for CSPs defined by \emph{pairwise-independent} predicates by combining Theorem \ref{thm:maincsp} with known integrality-gaps for such CSPs (\cite{BGMT12}). 

\PRnote{should the constants $c_{\epsilon}$ depend on $\epsilon$, don't we get a lower bound for some absolute constant in the exponent,  for Max3sat.}
\RMnote{This depends on whether or not the SA lower bound for MAX-3SAT has a dependence on $\epsilon$, right?}

\RMnote{Corollary for vertex cover?}
\RMnote{Also need to add proofs of these corollaries somewhere with appropriate citations to the $degree_+$/SA lower bounds.}

 The above results for CSPs are established through a more general claim on \emph{non-negative rank} of a class of matrices referred to as \emph{pattern matrices}. We explain this connection and results next.

\subsection{Lifting degree lower bounds to rank lower bounds}
In the seminal work introducing extended formulations, Yannakakis showed that the extended formulation complexity of an optimization problem is precisely the {\it non-negative rank} of the associated slack matrix.  In \cite{BFPS15}, this connection was subsequently extended to approximation by linear programs.  All known lower bounds on the size of extended formulations rely on this connection as do we.

\begin{definition}[Non-negative Rank]
Let $M$ be a non-negative matrix. The non-negative rank of $M$, denoted by $\nnr(M)$ is the least positive integer $r$ such that there exist non-negative rank $1$ matrices $M_1,\ldots, M_r$ such that $M = \sum_{i = 1}^r M_i$.
\end{definition}

Proving lower bounds on non-negative rank of specific matrices is often non-trivial; a significant breakthrough towards proving such lower bounds was achieved by the work of \cite{FMPTdW15} who showed a connection between communication complexity lower bounds and non-negative rank. 

We give a tight characterization of the non-negative rank of a broad-class of matrices--\emph{pattern matrices}--that were studied before in communication complexity \cite{RazM99,razborov,Sherstov11}. 

\begin{definition}[Pattern Matrix]  \label{def:lifted-matrix}
Fix positive integers $n$ and $q$.
Given functions $f: \on^n \to \R$ and $g : [q] \times [q] \to \on$, the composed function $f \circ g^{\otimes n} : [q]^n \times [q]^n \to \R$ is defined as,
\[ f\circ g^{\otimes n} (x,y) \defeq f \left(g(x_1,y_1),\ldots, g(x_n,y_n) \right) \mcom\]
where we have $x_i,y_i \in [q]$ for all $i \in [n]$.
The pattern matrix $M_f^g$ is the truth-table of the composed function $f \circ g^{\otimes n}$ expressed as a matrix, i.e., it is a matrix with rows and columns indexed by $[q]^n$ with,
\[M^{g}_f(x,y) \defeq f\circ g^{\otimes n} (x,y) \mper \]
\end{definition} 

The function $g : [q] \times [q] \to \on$ is referred to as the {gadget function}.  Throughout this work, we will use a slightly modified (in order to ensure balancedness) version of the Boolean \emph{inner-product function} as the gadget $g$.  Specifically we will set $q = 2^b$ for $b \in \N$, identify $[q]$ with $\zo^b$ and define 
$$\iproduct : \{0,1\}^b \times \{0,1\}^b \to \on \text{ given by } \iproduct(x,y) \defeq (-1)^{x_1 \oplus y_1} \cdot (-1)^{\oplus_{i=1}^b x_i y_i} \mper$$
With this choice of the gadget function $g$, we will use $M^b_f$ to denote the pattern matrix $M_f^g$; we also drop the superscript $b$ and use $M_f$ to denote $M_f^b$ when $b$ is clear from context. 

Our main result characterizes the non-negative rank of pattern matrices $M_f$ by a corresponding measure of $f$ that we define next.

\begin{definition}[Juntas and Non-negative Degree]
A function $h:\on^n \to \mathbb{R}$ is a \emph{$d$-junta} if it only depends on at most $d$ coordinates. A function $h:\on^n \to \rnng$ is a \emph{conical} $d$-junta if it can be written as a non-negative linear combination of non-negative $d$ juntas.

For any $f:\on^n \to \rnng$, the non-negative degree of $f$, written as $\degp(f)$, is the least positive integer $d$ such that $f$ can be written as a conical $d$-junta.
\end{definition}
\ignore{
\begin{definition}[Non-negative Degree]\label{deg:degp}
Left $f:\on^m \rightarrow \R_{\geq 0}$ be a non-negative function such that $\E[f] = 1$. The non-negative degree of $f$, written as $\degp(f)$, is the least positive integer $d$ such that $f$ can be written as a conical $d$-junta.
\end{definition}}

We show that for any non-negative function $f$, the non-negative rank of $M_f$ is essentially characterized by the non-negative degree of $f$. Indeed, it is easy to check that
\begin{equation}\label{eq:easyubnnr}
\nnr(M^b_f) \leq \binom{n}{\degp(f)} \cdot 2^{b \cdot \degp(f)}.
\end{equation}

We show a nearly matching lower bound for $\nnr(M^b_f)$; specifically, we show that if small positive shifts of $f$ have\footnote{Note that $\degp(f+\eta) \leq \degp(f)$ for all $\eta > 0$.} high non-negative degree, then $\nnr(M_f)$ is correspondingly large.


\begin{theorem}[$\nnr(M_f)$ vs $\degp(f)$]\label{thm:nnrlift}
There exist constants $c,C > 0$ such that the following holds. Let $f:\on^n \rightarrow \R_{\geq 0}$ be such that $\E[f] = 1$. Then, 
\[\nnr(M^b_f) \geq  2^{c \cdot b \cdot (\degp(f+\eta)-8\deg(f))}.\] \label{thm:main-tech}
for all $\eta \geq 1/n$ and $b \geq C(\log n)$.
\end{theorem}
Note that by \eref{eq:easyubnnr}, for $b > \log n$, $\nnr(M^b_f)  \leq 2^{2b \cdot \degp(f)}$. Thus in the interesting regime when $\deg_+(f) \gg \deg(f)$, the above theorem is tight up to constant factors (in the exponent) and working with $\degp(f+\eta)$.

\subsection{Previous work: Approximate non-negative rank versus non-negative rank}
The above result should be compared with similar results in \cite{GLMWZ15, LRS15} . Although they also obtain similar lifting theorems, a crucial difference is that they lower bound the \emph{approximate non-negative} rank of lifted matrices. For a non-negative matrix $M$, and $\eps > 0$, define the $\epsilon$-approximate non-negative rank as
$$\nnr^\epsilon(M) = \min\{\nnr(M'): \|M'-M\|_\infty \leq \epsilon \|M\|_\infty\}.$$

Clearly, $\nnr^\epsilon (M) \leq \nnr(M)$ for all $\epsilon > 0$. At a high-level, the previous works show lower bounds on $\nnr^\epsilon(M_f)$ (in terms of the \emph{approximate non-negative junta} degree of $f$). Similarly, while \cite{GLMWZ15} show a separation between $\nnr^\epsilon, \nnr^\delta$ for some constants $0 < \epsilon < \delta < 1$, the resulting matrices have large rank. Such lifting theorems are not enough to obtain our applications to CSPs -- \prettyref{thm:maincsp}, \prettyref{cor:maincsp} -- as matrices arising in these applications in fact have small approximate non-negative rank (roughly $n^{O(\log(1/\epsilon))}$) and small rank. This was one of the main reasons why the previous works only obtained quasi-polynomial size lower bounds.

In fact, before our work, the best separation between $\nnr^\epsilon$, $\mathsf{rank}$ and $\nnr$ was only quasi-polynomial. As a corollary of our results, we obtain weakly-exponential separation for an explicit matrix:
\begin{theorem}\label{thm:approxvsexact}
For all $\epsilon > 0$, there exist constants $0 < c_\epsilon, C_\epsilon$ such that the following holds. There exists an explicit non-negative matrix $M \in \rnng^{N \times N}$ such that $\mathsf{rank}(M), \nnr^\epsilon(M) \leq (\log N)^{C_\epsilon}$, and $\nnr(M) > N^{c_\epsilon}$.
\end{theorem}
\RMnote{Need to write down the proof somewhere.}

\subsection{Applications in Communication Complexity}
Analyzing lifted functions or pattern matrices has been a very useful tool in communication complexity over the last few years and our work builds on the techniques of \cite{GLMWZ15} who show lifting theorems for various \emph{rectangle-based} communication measures. Our main decomposition theorem, \prettyref{thm:maindecomp}, can be used to recover the main results of \cite{GLMWZ15}. Indeed, the main results of \cite{GLMWZ15} follow from a structural result about approximating \emph{rectangles} by \emph{juntas} -- an analogue of \prettyref{thm:mainjunta} that in turn follows easily from our decomposition theorem. For a more detailed comparison, see the discussion at the beginning of \prettyref{sec:decompalgorithm}. We believe that our decomposition theorem could lead to other such applications in future.
\section{Proof overview}

\subsection{Lifting $\degp$ lower bounds to non-negative rank}

The proof of \prettyref{thm:nnrlift} consists of two steps.  First, we show that if $\nnr(M_f)$ is small for a function $f$, then $f$ can be \emph{approximated} by a conical junta under a carefully chosen notion of {\it approximation}.  Second, we show that if $f$ can be so approximated by a conical junta, then $\degp(f+\eta)$ is small for $\eta \ll 1$.

Towards making this outline more precise, we begin by defining a notion of approximate conical juntas that plays an important role in our proofs. We first state some basic notations that we use throughout:

\begin{itemize}
\item For any function $f$, $\E[f]$ denotes the expectation of $f$ on the uniform distribution over its domain.
\item For any $x \in \on^n$ and $I \subseteq [n]$, we write $x_I$ to denote the projection of $x$ on to the coordinates in $I$.
\item A \emph{Boolean conjunction} $C:\on^n \to \rnng$ is defined by a subset $I \subseteq [n]$ of variables and $\alpha$ an assignment to the variables in $I$ by $C(x) = 2^{|I|} \cdot \1[x_I = \alpha]$. We say $C$ is a $d$-conjunction if $|I| \leq d$. Observe that we choose a non-standard scaling that satisfies $\E[C] = 1$.
\item For any $S \subseteq [n]$, the parity function $\chi_S(x) = \Pi_{i \in S} x_i$ for any $x \in \on^n$. Any function $f: \on^n \to \R$ has the \emph{Fourier expansion} $f(x) = \sum_{S \subseteq [n]} \hat{f}(S) \chi_S(x)$. The terms $\hat{f}(S)$ are the Fourier coefficients of $f$.
\end{itemize}
\begin{definition}[$\epsilon$-decaying functions]
For $0 < \epsilon < 1$, a function $h: \on^n \rightarrow \R$ is said to be $\epsilon$-decaying if $\E[h] = 0$ and for every $I \subseteq [n]$, $|\hat{h}(I)| \leq \epsilon^{|I|}.$
\end{definition}

\begin{definition}[$(\epsilon,\delta)$-approximate conical $d$-junta]
For $\epsilon,\delta \in (0,1)$, a function $f : \on^n \to \R_{\geq 0}$ with $\E[f] = 1$ is said to be an {\it $(\epsilon,\delta)$-approximate conical $d$-junta} if $f$ can be written as
\begin{equation}
f(z) = \sum_{i \in [N]} \lambda_i C_i(z) \cdot (1+ h_i(z)) + \gamma(z)
\end{equation}
for $d$-conjunctions $C_1,\ldots,C_N$, $\epsilon$-decaying functions $h_1,\ldots,h_N$,  $\lambda_1,\ldots, \lambda_N \in \rnng$ with $\sum_{i} \lambda_i \leq 1$, and a function $\gamma:\on^n \to \rnng$ such that $\E[\gamma] \leq \delta$. 
\end{definition}

Notice that the approximation by conical-juntas has two distinct error terms, the multiplicative errors due to the $\epsilon$-decaying functions $\{h_i\}$ and the additive error in the form of $\gamma$.

The first step in proving \prettyref{thm:nnrlift} is the following lemma saying that if $\nnr(M_f)$ is small, then $f$ is an approximate conical $d$-junta for small $d$.

\begin{lemma}[Non-negative rank to Approximate Conical Juntas] \torestate{\label{lem:approx-conical}
There exists a constant $\alpha_1 \geq 1$ such that the following holds. For $b \geq \alpha_1 \log n$, every function $f:\on^n \to \rnng$ with $\E[f] = 1$ is a $(2^{-b/2},2^{-bd/\alpha_1})$-approximate conical $d$-junta for all $d \geq \alpha_1 (\log \nnr(M_f^b) )/b$.
}
\end{lemma}

\ignore{There exists a constant $\alpha_1 > 0$ such that the following holds:
Suppose $r = \nnr(M_f)$, then $f$ is a $(\epsilon,\delta)$-approximate conical $d$-junta with $\epsilon = 2^{-0.5 b}$, $\delta = r \cdot 2^{-bd}$ and  $d  = \frac{\alpha_1 \cdot \log{r}}{b}$.
\end{lemma}}

We defer the sketch of the proof of the lemma to the next section and continue with our outline of the proof of \prettyref{thm:nnrlift}.

Given the above lemma, the final step in proving \prettyref{thm:nnrlift} is to show a connection between $\degp(f)$ and $(\epsilon,\delta)$-approximation by conical juntas.  Specifically, we show a certain {\it robustness} of the class of conical juntas: if a function $f$ is an $(\epsilon,\delta)$-approximate conical $d$-junta for sufficiently small $\epsilon$ and $\delta$, then the function $f+\eta$ is an exact conical $8d$-junta for a small constant $\eta$.

\begin{lemma} 
\torestate{
\label{lem:approx-to-exact}
Suppose $f : \on^n \to \rnng$ with $E[f] \leq 1$ is an $(\epsilon,\delta)$-approximate conical $d$-junta for $\epsilon < 1/n^4$ and some $d \geq \deg(f)$ and $\delta < 1/n^{8d}$ then \[ \degp\left(f+\frac{1}{n}\right) \leq 8d\]
}
\end{lemma}

\prettyref{lem:approx-conical} and \prettyref{lem:approx-to-exact} together imply \prettyref{thm:nnrlift} almost immediately by setting the parameters appropriately; see \prettyref{sec:junta-rect-nnr}.
%

We defer the proof of \prettyref{lem:approx-to-exact} to \prettyref{sec:approx-to-exact}.  In what follows, we sketch the key ideas underlying the proof of \prettyref{lem:approx-conical}.



\subsection{Approximating Rectangles by Conical Juntas}
We now sketch the proof of \prettyref{lem:approx-conical}. To do so, we need the following basic definition\footnote{We work with densities (instead of equivalently working with probability density functions or just non-negative functions) as  keeping track of errors is cleaner under this normalization.}.
\begin{definition}[Density]
A function $p: [q]^{n} \rightarrow \R_{\geq 0}$ is said to be a density if $\E[ p(x) ] = 1.$ A density $p$ defines a corresponding random variable $X$ on $[q]^n$ where $Pr[X= x] = p(x) \cdot q^{-n}$. We denote $X \sim p$ this random variable.
\end{definition}

Recall the statement of the lemma: we have a density $f$ on $\on^n$ such that $M_f$ has small non-negative rank and we want to show that $f$ is a low-degree approximate conical junta. Let $\nnr(M_f) = r$.  By definition, $M_f = \sum_{i \in [r]} M_i$ where each $M_i$ is a non-negative rank one matrix; further, by appropriate normalization, we can assume that $M_i = \lambda_i u_i v_i^\dagger$, where $u_i,v_i$ are densities on $[q]^n$ and $\lambda_i > 0$. This decomposition of the matrix $M_f$ into non-negative rank $1$ matrices $\{M_i\}$, yields a corresponding decomposition of the function $f$ into a sum of non-negative functions, one corresponding to each rank one matrix $M_i$.

Formally, let us denote by $G :[q]^n \times [q]^n \to \on$ the function $G \seteq g^{\otimes n}$.  By definition, the entries of the matrix $M_f$ are given by $M_f(x,y) = f\left( G(x,y) \right)$.  For $z \in \on^n$, let $(X,Y) \sim  G^{-1}(z)$ denote a uniformly random pair chosen from the set of pairs $G^{-1}(z) \subseteq [q]^n \times [q]^n$.  With this notation,
\begin{equation}\label{eq:intromainjunta1}
 f(z) = \E_{(X,Y) \sim G^{-1}(z)} M_{f}(X,Y) = \sum_{i \in [r]} \E_{(X,Y) \sim G^{-1}(z)} \left[ M_i(X,Y) \right] = \sum_{i \in [r]} \lambda_i \E_{(X,Y) \sim G^{-1}(z)}\sbkets{u_i(X) v_i(Y)}\mper
 \end{equation}

Borrowing terminology from communication complexity, we will refer to the rank one matrices $u_i v_i^\dagger$ as {\it rectangles}.  In order to approximate the function $f$ by a conical junta, it suffices to approximate the terms corresponding to each {\it rectangle} by a conical junta.  We exhibit such an approximation for all {\it large rectangles}.

\ignore{

 For the sake of concreteness, let us suppose
\[M_i(x,y) = \lambda_i \cdot u_i(x) \cdot v_i(y)\]
where $u_i, v_i : [q]^n \to \R_{\geq 0}$ are densities.

Then, from the above arguments, we get
$$f(z) = \sum_{i=1}^r \lambda_i \E_{(x,y) \sim G^{-1}(z)}[u_i(x) v_i(y)].$$}

Towards this end, for two densities $u,v$ on $[q]^n$, define $Acc_{u,v}:\on^n \to \rnng$ by
\begin{equation}\label{eq:defacc}
Acc_{u,v}(z) = \E_{(X,Y) \sim G^{-1}(z)}[u(X) v(Y)].
\end{equation}

Note that $\E[Acc] = \E[u(x) v(y)] = 1$. Thus, $Acc$ is a density on $\on^n$. Indeed, it is easy to check that $Acc_{u,v}$ is the density of the random variable $G(X,Y)$ for $X \sim u$ and $Y \sim v$ ($X,Y$ independent). Using this definition in \prettyref{eq:intromainjunta1}, we get
$$f(z) = \sum_{i=1}^r \lambda_i \cdot Acc_{u_i,v_i}(z).$$

This motivates the study of functions $Acc_{u,v}$ for rectangles. Indeed, structural results characterizing such functions form the core of previous results on pattern matrices \cite{GLMWZ15,Sherstov11}. We show that functions $Acc_{u,v}$ as above are \emph{simple} when the rectangle $u\times v$ is \emph{large}. To formalize this we need the notion of \emph{min-entropy}.

\begin{definition}[Min-Entropy]
For a density $u$ on $[q]^{n}$, the \emph{min-entropy} of $u$, $H_\infty(u)$, is defined by\footnote{Note that this is the same as the more standard definition of $\min_{x \in \on^n} \log(1/\pr[X=x])$ where $X \sim u$.}\footnote{Throughout this work, all logarithms are to the base $2$.} $$H_{\infty}(u) = \min_{x \in \on^n} \log {(q^n/u(x))}.$$
\end{definition}

For intuition, it is helpful to think of the special case where the densities $u_i, v_i$ correspond to uniform distributions over some subsets $U_i, V_i$ of $[q]^n$ respectively.  The rectangle $M_i$ is said to be large, if the sets $U_i$ and $V_i$ are both {\it large}, of size at least $q^{n}/2^{C}$ for $C \ll n$. More generally, the rectangle $M_i$ is large if the distributions $u_i,v_i$ each have min-entropy at least $n \log{q} - C$. We will refer to $C$ as the {\it min-entropy deficiency}.

Since $M$ is the sum of $r$ rectangles, one can argue that it is approximated by {\it large} rectangles whose min-entropy deficiency is at most $O(\log{r})$.  The contribution from all the {\it small} rectangles can be included into the additive error term $\gamma(z)$ in the approximation for $f$.  The main work lies in showing that every {\it large rectangle} is approximated by conical juntas.  

\begin{theorem}[Junta Approximation]
\torestate{
\label{thm:mainjunta}
There exists a constant $\alpha_2 \geq 1$ such that the following holds. Let $u,v$ be densities over $[q]^n$ with $q =2^b$ such that $H_{\infty}(u)+ H_{\infty}(v) \geq 2b(n-t)$. Then, for all $b \geq \alpha_2 \log n$ and $d \geq \alpha_2 t$, $Acc_{u,v} \on^n \to \rnng$ is a $(2^{-0.5b},(2^{-0.5b})^d)$-approximate conical $d$-junta.}
\end{theorem}

\cite{GLMWZ15} also show a similar, but weaker, junta approximation theorem. In the present context, they essentially show that $Acc_{u,v}$ can be approximated as $Acc_{u,v}(z) = (1\pm 2^{-\Omega(b)}) \cdot h(z) \pm 2^{-\Omega(bd)}$ where $h$ is a conical $d$-junta. Note that the multiplicative error is only of the order $2^{-\Omega(b)}$ and this was a critical bottleneck in using their results to prove a lifting theorem for non-negative rank as in \prettyref{thm:nnrlift} (instead of for approximate non-negative rank).  In comparison, we get exponentially small error in terms of approximate conical $d$-juntas. The latter is in fact stronger; a straightforward extension of our arguments can in fact recover the corresponding statement of \cite{GLMWZ15}.


\subsection{Decomposing High-Entropy Distributions}

The proof of \prettyref{thm:mainjunta} relies on a crucial decomposition lemma for high-entropy distributions that may be of independent interest. Let $u,v$ be two densities over $[q]^n$ with min-entropy at least $(n-C) \cdot \log{q}$ for some $C \ll n$ and let $X \sim u, Y \sim v$ be sampled independently.



A particularly simple class of high min-entropy distributions are those where a subset of $C$ coordinates of $X$ are fixed, while the rest are uniformly random. That is, for some set $I \subseteq [n]$ with $|I| \leq C$, $X_I$ is a fixed string whereas $X_{[n]\setminus I}$ is uniformly random over $[q]^{[n] \setminus I}$. Similarly, $Y$ could satisfy a similar property for a set $J \subseteq [n]$ with $|J| \leq C$. An especially desirable scenario is one where $X,Y$ are \emph{aligned} in the sense that $I=J$. For such aligned distributions, the random variable $Z = g^n(X,Y) \in \on^n$ is such that $Z_I$ is fixed while $Z_{\overline{I}}$ is uniformly random.  In other words, the probability density of $Z$ is a $C$-junta depending only on $I$.

We will show that as long as $X,Y$ have high min-entropy, the product distribution $X \times Y$ can be decomposed into distributions that are essentially as simple and aligned as in the above discussion.

To this end, we next introduce the notion of \emph{blockwise-dense} distributions; they were first defined in \cite{GLMWZ15} and play a crucial role here.
\begin{definition}
A distribution $X$ on $[q]^{n}$ is \emph{\bd} if for every $I \subseteq [n]$, $H_\infty(X_I) \geq 0.8 \cdot \log q \cdot |I|$. We say a density $u$ on $[q]^n$ is $\bd$ if $X \sim u$ is $\bd$.
\end{definition}

\PRnote{use blocks or coordinates over q???}
\RMnote{I think coordinates is better when dealing with $q$.}
\begin{definition}
A distribution $X$ on $[q]^{n}$ is a $d$-\cbd (``conjunctive blockwise-dense'') distribution if for some set of coordinates $I \subseteq [n]$, $|I| \leq d$, $H_\infty(X_I) = 0$ and for every $J \subseteq [n] \setminus I$, $H_\infty(X_J) \geq 0.8 \cdot \log{q} \cdot |J|$. We refer to $d$ as the \emph{degree} of the \cbd distribution, and the set of blocks $I$ as the {\it fixed} blocks. We say two $d$-CBD distributions $X,Y$ on $[q]^{n}$ are aligned if the \emph{fixed} blocks $I$ are the same in both.

Analogously, we say two densities $u,v$ over $[q]^n$ are aligned $d$-\cbd if the random variables $X,Y$ are aligned $d$-\cbd distributions for $X \sim u, Y \sim v$.
\end{definition}

The technical core of our results is the following lemma stating that any two independent high-entropy densities $u,v$ over $[q]^n$ can be approximated by a convex combination of aligned $d$-\cbd densities for small $d$. The error of the approximation will depend on the entropy deficiency of $u  \otimes v$ and the degree of the $\cbd$ distributions used in the approximation.


\begin{theorem}\label{thm:maindecomp}
There exists a constant $c \geq 1$ such that the following holds. For $n \geq 1$ and $q \geq n^c$, let $u,v$ be two densities on $[q]^n$ with $H_\infty(u) + H_\infty(v) \geq 2(n-t) \cdot \log{q}$. Then, for all $d \geq ct/(\log q)$, the product density $u \otimes v$ on $[q]^n \times [q]^n$ can be written as a convex combination of densities $u_1 \otimes v_1, u_2 \otimes v_2,\ldots, u_N \otimes v_N$, and $\gamma_{err}$, i.e., 
$u \otimes v = \sum_{i=1}^N \lambda_i u_i \otimes v_i + \lambda_{err} \gamma_{err}$, such that
\begin{itemize}
\item $0 \leq \lambda_1,\ldots,\lambda_N, \lambda_{err} \leq 1$, $\sum_{i=1}^N \lambda_i + \lambda_{err} = 1$.
\item $|\lambda_{err}| <  q^{-\Omega(d)}$.
\item For every $i \in [N]$, $X_i \sim u_i$, $Y_i \sim v_i$ are aligned $d$-\cbd distributions.
\end{itemize}
\end{theorem}

\ignore{
\begin{theorem}\label{thm:maindecomp}
Let $X,Y$ be two independent distributions over $[q]^{n}$ with $H_\infty(X) + H_\infty(Y) \geq 2n \cdot \log{q} - t$. Then, $X\times Y$ can be written as a convex combination of distributions $X_1 \times Y_1,\ldots,X_N \times Y_N , E$ over $[q]^{n} \times [q]^{n}$, i.e., $X \times Y = \sum_{i=1}^N \lambda_i X_i \times Y_i + \lambda_{err} E$, such that
\begin{itemize}
\item $0 \leq \lambda_1,\ldots,\lambda_N, \lambda_{err} \leq 1$, $\sum_{i=1}^N \lambda_i + \lambda_{err} = 1$.
\item $|\lambda_{err}| < 2^{t} (d q^{-0.05})^d$.
\item For every $i \in [N]$, $X_i, Y_i$ are aligned $d$-\cbd distributions.
\end{itemize}
\end{theorem}}

\prettyref{thm:mainjunta} follows easily from the above using some ``extractor''-like properties (cf.~Fact \ref{fact:ipextractor}) of the slightly modified inner-product function $\iproduct$. We defer the details of the proof of the theorem to the corresponding section.

\subsection{Organization}
We present the proof in a top-down manner: We first prove \prettyref{thm:nnrlift} assuming \prettyref{thm:mainjunta}. We then prove \prettyref{thm:mainjunta} assuming \prettyref{thm:maindecomp} (this is almost immediate). Finally, we prove \prettyref{thm:maindecomp}. We then prove \prettyref{thm:maincsp}, \prettyref{cor:maincsp} in \prettyref{sec:nnrlift}.

\section{Preliminaries}
We describe some basic notation that we use throughout\footnote{To have all notations together, some are repeated from the introduction.}.
\subsection{Basic Notation}
\begin{enumerate}
\item $\P_d^n$ denotes the collection of all polynomials of degree at most $d$ on $n$ variables on $\on^n$.
\item $\1(E)$ is the indicator for the event $E$ normalized to have mean $1$. That is, $\1(E)$ is $0$ when $E$ doesn't happen and $1/\Pr[E]$ when $E$ happens.
\item For any function $f$, $\E[f]$ denotes the expectation of $f$ on the uniform distribution over its domain.
\item For matrices $M$, $\E[M]$ denotes the expectation of $M(x,y)$ under $x,y$ being uniformly random indices for its rows and columns.
\item For any $x \in \on^n$ and $I \subseteq [n]$, we write $x_I$ to denote the projection of $x$ on to the coordinates in $I$.
\item A \emph{Boolean conjunction} $C:\on^n \to \rnng$ is defined by a subset $I \subseteq [n]$ of variables and $\alpha$ an assignment to the variables in $I$ by $C(x) = \1[x_I = \alpha]$. We say $C$ is a $d$-conjunction if $|I| \leq d$. Observe that we choose a non-standard scaling that satisfies $\E[C] = 1$.
\item For any $S \subseteq [n]$, the parity function $\chi_S(x) = \Pi_{i \in S} x_i$ for any $x \in \on^n$. Any function $f: \on^n \to \R$ has a Fourier expansion: $f(x) = \sum_{S \subseteq [n]} \hat{f}(S) \chi_S(x)$. The terms $\hat{f}(S)$ are the Fourier coefficients of $f$.

\end{enumerate}

\subsection{Sherali-Adams Linear Programming Relaxations}
Our results relate arbitrary linear programming relaxations for CSPs to the Sherali-Adams hierarchy. We discuss the latter class of linear programs next. We begin with the definition of a {\it degree $d$ pseudo-expectation}


\begin{definition}[Sherali-Adams Pseudoexpectation] \label{def:pE}
A degree $d$ Sherali-Adams pseudoexpectation, $\pE$, is a linear operator on the space of degree at most $d$ polynomials, $\P_d^n$, such that \begin{enumerate}
\item For every non-negative $p \in \P_d^n$ that depends on only $d$ variables, $\pE[p] \geq 0,$ and
\item $\pE[\1] = 1.$
\end{enumerate}
\end{definition}
Since $\pE$ is a linear, it is completely specified by its values on multilinear polynomials, in particular by the values $\pE[\chi_{S}(x)]$ for $S \subseteq [n]$, $|S| \leq d.$

Sherali-Adams linear programming relaxations can be equivalently described using a collection of probability distributions over local assignments. The above view is more convenient for us. We refer the reader to \cite{CLRS13} for a detailed discussion.


The degree $d$-Sherali-Adams linear programming relaxation for a CSP solves the following optimization problem. Given an instance $\I$ of a $k$-ary CSP, we can canonically encode it as a polynomial of degree $k$ $P_\I:\on^n \to [0,1]$ such that $P_\I(x) = \I(x)$ for all assignments $x \in \on^n$. Then, the degree $d$-Sherali-Adams relaxation is

 \begin{equation} \max_{\pE} \pE[ P_\I(x)], \label{eq:SA-def} \end{equation}
where $\pE$ ranges over all degree $d$ Sherali-Adams pseudoexpectations. We define $\SA_d(\I)$ as the value of the the optimization problem \eqref{eq:SA-def}.
The above optimization problem can be solved using a linear program on $n^{O(d)}$ variables and constraints. Note that $\opt(\I) \leq \SA_d(\I)$.

\paragraph{Sherali-Adams LP and Non-negative Degree:}
Linear programming duality gives an elegant characterization of the performance of Sherali-Adams LP on a CSP in terms of non-negative degree.

\begin{fact}[Sherali-Adams value and Non-negative Degree \cite{CLRS13}] \label{fact:savaldegp}
Let $P:\on^k \to \zo$ be a predicate and $\I$ be an instance of $CSP(P)$. Then, $\SA_d(\I) \leq c$ if and only if $\degp(c-\I) \leq d$.
\end{fact}

\section{Juntas, Rectangles, and Non-negative Rank of Lifted Matrices} \label{sec:junta-rect-nnr}

In this section, we will show our main \prettyref{thm:nnrlift} assuming \prettyref{thm:mainjunta}. Let $f:\on^n \to \rnng$ be as in the theorem with $\E[f] = 1$, and let the gadget $g:[q] \times [q] \to \on$ and $b$ be as in the theorem. We will show that if $\nnr(M_f)$ is small, then $\degp(f+\eta)$ is small, where $\eta = O(1/n)$. Concretely, given a small rank non-negative factorization of $M_f$, we use the factorization to get a small-degree conical junta approximating $f+\eta$. As described in the introduction, this is done in two modular steps: Lemmas \ref{lem:approx-conical}, and \ref{lem:approx-to-exact}.  First, we show how \prettyref{lem:approx-conical} and \prettyref{lem:approx-to-exact} together immediately imply \prettyref{thm:nnrlift}.

\begin{proof}[Proof of \prettyref{thm:nnrlift}]
Fix a constant $C$ such that $C \geq \max(16\alpha_1, 1000)$ for $\alpha_1$ from \prettyref{lem:approx-conical}.  
Let $R \seteq \nnr(M_f^b)$.  By \prettyref{lem:approx-conical}, this implies that $f$ is a $(2^{-b/2}, 2^{-bd/\alpha_1})$-approximate conical d-junta for $d \geq \alpha_1\log R/b$.  For $b \geq C \log n$, $2^{-b} \leq \min(1/n^{1000}, 1/n^{16 \alpha_1})$.  Hence, $f$ is a $(1/n^{500}, 1/n^{16d})$-approximate conical $d$-junta with $d = \max\{\lceil\alpha_1 \log R/b \rceil, \deg(f)\}$.  By \prettyref{lem:approx-to-exact}, this implies that 
\[ \degp\left(f + \frac{1}{n}\right) \leq 8 \cdot \frac{\alpha_1 \log R}{b} + 8\deg(f) \mcom\]
which yields the inequality, 
\[ R \geq 2^{\Omega(b) \cdot (\degp(f+\nfrac{1}{n})-8\deg(f))} \mper\]
\end{proof}

For the rest of this section we adopt the following assumptions:
\begin{center}
  \fbox{\begin{minipage}{\textwidth}
  \paragraph{Important Parameters}
 \begin{itemize}
 	\item $f:\on^n \rightarrow \rnng$ with $\E[f] = 1$.
    \item The \emph{block-length} of the gadget $b = C \log n$ for a sufficiently large constant $C$. 
    \item Recall that the gadget is defined at any $x,y \in \zo^b$ by: 
    \[
    \iproduct(x,y) \defeq (-1)^{x_1 \oplus y_1} \cdot (-1)^{\oplus_{i=1}^b x_i y_i} \mper\]
\end{itemize}

  \end{minipage}}
\end{center}

\subsection{Approximation by Approximate Conical Juntas}

Here we prove \prettyref{lem:approx-conical} which we restate for convenience.

\restatelemma{lem:approx-conical}


\begin{proof}
For the sake of brevity, let us set $G \seteq g^{\otimes n}$ and $R \seteq \nnr(M_f^b)$.  For every $z \in \on^n$, and $(x,y) \in G^{-1}(z)$, we have $M_f(x,y) = f(z)$. The high-level idea is as follows. From the definition of $M_f$, we have
 $$ f(z) =  \E_{(X,Y) \sim G^{-1}(z)}[ M_f(X,Y)].$$
Further, by definition of $\nnr(M_f)$, the matrix $M_f$ can be expressed as a sum of $\nnr(M_f)$ {\it non-negative rank-1 matrices}.  In turn, this yields a decomposition of $f$ into a sum of a family of non-negative functions. We then use  \prettyref{thm:mainjunta} to approximate each of these functions by approximate conical juntas, thereby yielding the desired approximation for $f$.

Concretely, from the definition of non-negative rank, there exists a collection of densities on $\on^{bn}$, $\{u_i \mid 1 \leq i \leq R\}$ and $\{v_i \mid 1 \leq i \leq R\}$, and a set of non-negative constants $\lambda_1, \lambda_2, \ldots, \lambda_R$ such that $$M_f = \sum_{i = 1}^R \lambda_i u_i v_i^{\dagger}.$$
Observe that $$ \sum_{i = 1}^R \lambda_i = \sum_{i = 1}^R \lambda_i \E[ u_iv_i^{\dagger}]= \E[M_f] = 1.$$

Now, for any $z \in \on^n$,
\begin{align*}
f(z) &= \E_{(X,Y) \in G^{-1}(z)}[M_f(X,Y)]\\
&= \E_{(X,Y) \in G^{-1}(z)}\sbkets{\sum_{i=1}^R \lambda_i u_i(X) v_i(Y)}\\
&= \sum_{i=1}^R \lambda_i \cdot \E_{(X,Y) \in G^{-1}(z)}[u_i(X) v_i(Y)]\\
&= \sum_{i=1}^R \lambda_i \cdot Acc_{u_i,v_i}(z).
\end{align*}
(Recall the definition of $Acc$ from Equation \prettyref{eq:defacc}.)
\ignore{where we define $\1_z(x,y):[q]^n \times [q]^n \to \rnng$ by $\1_z(x,y) = (q^{2n}/|G^{-1}(z)|) \cdot \1(G(x,y) = z)$. Note that $\1_z$ is a density and in particular $\E[\1_z] = 1$. }

In analogy with communication complexity, we will refer to the rank $1$ matrices $u_i v_i^\dagger$ as {\it rectangles}. We will split the family of rectangles in to {\it large} and {\it small}.  To this end, fix $t = 4 \log {(R)}$.
A rectangle $u_i v_i^{\dagger}$ will be referred to as {\it large}, if the min-entropies of $u_i$ and $v_i$ are large.  More precisely, let
$$Q = \{i \in R\mid H_\infty(u_i) + H_\infty(v_i) \geq 2 (n-t) (\log q)\}.$$
We can now write $f$ as a sum $f = J + \delta_1$ where,
\[J(z) = \sum_{i \in Q} \lambda_i Acc_{u_i,v_i}(z)\] and
\[ \delta_1(z) = \sum_{i \not \in Q} \lambda_i Acc_{u_i,v_i}(z) .\]


Now, for each $i \in Q$, by \prettyref{thm:mainjunta} applied to $u_i,v_i$, $Acc_{u_i,v_i}$ is a $(\epsilon,\epsilon^d)$-approximate conical $d$-junta for all $d \geq \alpha_2 t$. Therefore, $J$ is an $(\epsilon,\delta')$-approximate conical $d$-junta with $\delta' = \left(\sum_{i \in Q} \lambda_i \right) \cdot \epsilon^d \leq \epsilon^d$. 



Now we will bound the total additive error due to the small rectangles.  Observe that for any $i \not \in Q$, $\lambda_i \leq 2^{-t/2}$. This is because, $\lambda_i \E_{y}[ u_i(x) v_i(y)] = \lambda_i u_i(x) \leq \E_y[ M_f(x,y)] = \E[f] = 1$ and similarly, $\lambda_i v_i(y) \leq 1$ for any $y$. Further, recall that $Acc_{u,v}$ is a density for all densities $u,v$. Thus,
$$\E[\sum_{i \notin Q} \lambda_i A_{u_i,v_i}] \leq \sum_{i \notin Q} 2^{-t/2} \leq 2^{-t/2} R.$$
\ignore{
Now, $\E_{z} \E[ \1_z(x,y) u_i(x) v_i(y)] = \E_{x,y}[ u_i(x) v_i(y)] = 1.$ Thus, $\E[\delta_1] = \sum_{i \not \in Q} \lambda_i \leq 2^{-t/2} R.$ }

Therefore, $f = J+\delta_1$ is an $(\epsilon,\epsilon^d + 2^{-t/2}R)$-approximate conical $d$-junta for all $d \geq \alpha_2 t/b$. Choosing $t = 4 \log R$ and $d = \max\{4\alpha_2 \log R/b, 2\log  R/b\}$ proves the lemma.
%
\end{proof}

\subsection{Approximate Conical Juntas to Conical Juntas} \label{sec:approx-to-exact}

\begin{center}
  \fbox{\begin{minipage}{\textwidth}
  \paragraph{Notation}
 \begin{itemize}
 	\item $\cC_{\leq D}$: cone of non-negative $D$-juntas on $\on^{n}.$
 	\item $\cL: \on^n \rightarrow \R$: separating function.
\end{itemize}
  \end{minipage}}
\end{center}



In this section we will prove \prettyref{lem:approx-to-exact} which asserts that if $f$ is a low-degree approximate conical junta, then $f+\eta$ is a low-degree conical junta for $\eta$ sufficiently small.


\restatelemma{lem:approx-to-exact}

At a high-level the proof is as follows. Suppose for the sake of contradiction that $\degp(f+1/n) \geq 8d$. Then, there is a \emph{nice separating} functional $\L$ such that $\iprod{\L,f} \leq -1/n$, and $\iprod{\L,h} \geq 0$ for all $h \in \cC_{\leq 8d}$. We then use further properties of the functional that the latter property implies $\iprod{\L,h} > -1/n$ for all $(\epsilon,\epsilon^d)$-approximate conical $d$-junta - leading to a contradiction.

We first develop the requisite technical machinery concerning conical juntas and separating functionals.




\begin{lemma} \label{lem:sadual}
Suppose $\deg(f) \leq D <  \degp(f+\eta).$ There exists a degree $D$ function $\L:\on^{n} \rightarrow \R$ such that:
\begin{enumerate}
\item $\E[ \L ] = \E[ \L \cdot 1] = 1$.
\item $\E[ \L f] < -\eta$.
\item $\E[ \L h] \geq 0$ for every conical $D$-junta $h$ on $\on^n.$
\item $|\hat{\L}(S)| \leq 1$ for every $|S| \leq D$.
\item $\norm{\L}_{\infty} \leq n^{D}$
\end{enumerate}
\end{lemma}
\begin{proof}
Observe that the set of conical $\leq D$-juntas denoted by $\cC_{\leq D}$ is convex. On the other hand from the hypothesis, we have that $f +\eta \not \in \cC_{\leq D}$. Thus,
there exists a function $\L:\on^n \rightarrow \R$ such that $\langle \L, h \rangle \geq 0$ for every $h \in \cC_{\leq d}$ but $\langle \L , (f+\eta) \rangle < 0$.  Moreover, since $\cC_{\leq D}$ and $f+\eta$ are contained in the linear subspace of degree $D$ polynomials, without loss of generality, we can assume that $\L$ is also a degree $D$ polynomial.

The first three properties are simple to verify. Since the constant function $\1 \in \cC_{\leq D}$, we can assume (by rescaling, if needed) that $\langle \L , 1 \rangle = 1$ giving us the first property. Further, since $\langle \L , f \rangle = \L \cdot (f+\eta) - \langle \L , \eta \rangle \leq -\eta$ giving us the second property. The third property follows from our definition of $\L$.


We next bound the Fourier coefficients of $\L$. First observe that for any $S \subseteq [n]$, $|S| \leq D$, $\1+\chi_S$ is a non-negative $D$-junta. Therefore, $\iprod{\L,1+\chi_S} \geq 0$ so that $\iprod{\L,\chi_S} \geq -\iprod{\L,1} = -1$. Similarly, $\iprod{\L,1-\chi_S} \geq 0$ so that $\iprod{\L,\chi_S} \leq 1$. Thus, $|\hat{\L}(S)| = |\iprod{L,\chi_S}| \leq 1$. 
\ignore{\begin{align*}
\left|\hat{L}(S) \right| = \left| \langle \L , \chi_{S}(x) \rangle \right| = \left|{\sum_{\alpha \in \{-1,1\}^S} \langle \L , \chi_{S}(\alpha) \cdot
  \1[x_{S } = \alpha]} \rangle \right| \leq \sum_{\alpha \in \{-1,1\}^S} \langle \L , \1[x_{S} = \alpha] \rangle = 1
  \end{align*}}


The final property follows as $\norm{L}_{\infty} = \norm{\sum_{S, |S| \leq D} \hat{\L}(S) \chi_S(x)} \leq \sum_{|S| \leq D} | \hat{\L}(S) | \leq n^{D}$.
\end{proof}

The following technical property of $\L$ constructed in Lemma \ref{lem:sadual} will be required in our proof.

\begin{lemma}
Let $\L:\on^n \to \R$ be the separating function of degree $D$ given by Lemma \ref{lem:sadual}.  Let $h:\on^n \to \rnng$ be a  non-negative junta that depends only on variables $T \subseteq [n]$ and let $S$ be any subset of $[n]$ such that and $|S|+ |T| \leq D$. Then, $$\E[ \L h \chi_S] \leq \E[ \L h].$$ \label{lem:conditioning}
\end{lemma}

\begin{proof}
Note that $h(1-\chi_S)$ is a non-negative $D$-junta. Therefore, $\iprod{\L,h(1-\chi_S)} \geq 0$; the claim follows.
\ignore{
For $\1_{\chi_S(x) = 1}$ and $\1_{\chi_S(x) = -1}$, the $0$-$1$ indicators of the sets $\{ x \mid  \chi_S(x) = 1\}$ and $\{ x \mid \chi_S(x) = -1\}$ respectively, we have:
\begin{equation} \E[ \L c \chi_S]  = \E[ \L c \1_{\chi_S(x) = 1}] - \E[ \L c \1_{\chi_S(x) = -1}] \label{eq:temp1}
 \end{equation}
Now, $c \cdot \1_{\chi_S(x) = -1}$ is a conical $D$-junta from the assumption in the statement. Thus, $\E[ \L c \1_{\chi_S(x) = -1}] \geq 0$. Thus, \eqref{eq:temp1} yields:
$$\E[ \L c \chi_S]  \leq \E[ \L c \1_{\chi_S(x) = 1}] + \E[ \L c \1_{\chi_S(x) = -1}] = \E[ \L c].$$}
\end{proof}

We are now ready to prove Lemma \ref{lem:main-estimate} which can be seen as a robust version of the property that $\E[\cL h] \geq 0$ for every  non-negative $D$-junta $h$.

\begin{lemma} \label{lem:main-estimate}
Let $\L$ be a separating function of degree $D=4d$ as in Lemma \ref{lem:sadual}. Then, for any non-negative $d$-junta $c$ with $\|c\|_\infty \leq 1$, and any $(1/n^4)$-decaying function $h$, $\E[ \L c (1+h)] \geq -n^{-8d} $.
\end{lemma}

\begin{proof}
Write $h = h_{low} + h_{high}$ where $h_{low} = \sum_{|S| \leq D-d} \hat{h}(S) \chi_S,$ and $h_{high} = \sum_{|S| > D-d} \hat{h}(S) \chi_S.$

We have: $$\E[ \L c (1 + h)]  = \E[\L c (1 + h_{low})] + \E[ \L c h_{high}].$$
Let $\epsilon = 1/n^4$. Now, $\L c$ is of degree at most $(D+d)$ and $\|\L c\|_\infty \leq n^D$. Further, for any $x$,
\begin{align*}
|h_{high}(x)| &= |\sum_{S: |S| \geq D-d} \hat{h}(S) \chi_S(x)| \leq \sum_{S: |S| \geq D-d} \epsilon^{|S|}\\
&\leq \sum_{\ell=D-d}^n \epsilon^\ell \cdot n^\ell\\
&\leq 2(\epsilon n)^{D-d},
\end{align*}
where the last inequality follows as $\epsilon n = 1/n^3 < 1/2$. Therefore, $\|h_{high}\|_\infty \leq 2 (\epsilon n)^{D-d}$ and 

\begin{equation}
|\E[ \L c h_{high}]| \leq \|\L c\|_\infty \cdot \|h_{high}\|_\infty \leq 2 n^{2D-d} \cdot \epsilon^{D-d}. \label{eq:part1}
\end{equation}

Next, note that $h_{low}$ is a linear combination of parities of degree at most $d$ and that $c$ is a function of degree at most $d$. Thus, by \prettyref{lem:conditioning}, we have,
\begin{align*}
|\E[ \L c h_{low}]| &\leq \sum_{|S| \leq D-d} |\hat{h}(S)| | \E[ \L c \chi_S]|\\
&\leq \sum_{|S| \leq D-d} |\hat{h}(S)| \E[ \L c].
\end{align*}

Since by definition $\E[h] = 0$, we have:
\begin{equation}
\E[ \L c (1 + h_{low})] \geq \E[ \L c ] (1 - \sum_{1 \leq |S| \leq D-d} \epsilon^{|S|}) \geq  \E[\cL c] (1- 2D n \epsilon). \label{eq:part2}
\end{equation}
Using that $\epsilon < 1/n^{4}$ and $D < n$, we have $E[ \L c (1 + h_{low})] \geq 0$. Using \eqref{eq:part1} and \eqref{eq:part2}, $$\E[\L c (1+h)] \geq - 2 \epsilon^{3d}n^{7d} \geq - n^{-8d}.$$
\end{proof}

Finally, we can complete the proof of \prettyref{lem:approx-to-exact}.
\begin{proof}[Proof of \prettyref{lem:approx-to-exact}]
Fix $\eta = \frac{1}{n}.$ For the sake of contradiction, assume that $\degp(f+\eta) \geq 8d$.  Consider the functional $\cL$ given by \prettyref{lem:sadual} with $D = 4d$.

Since $f$ is an $(\epsilon,\delta)$-approximate $d$-junta, we have
\[ f(z) = \sum_{i} \lambda_i c_i(z) (1 + h_i(z)) + \gamma(z) \]
where the functions $h_i$ are $\epsilon$-decaying and the function $\gamma$ satisfies $\E[|\gamma|] \leq \delta$.
Now, take inner products with $\cL$ on both sides of the above equation.  On one side, we get \[\E[ \cL \cdot f] \leq -\eta \mper \]
On the other side, we get
\begin{align*}
\sum_{i} \lambda_i \E[ \cL c_i (1+ h_i)] + \E[\cL \gamma]
& \geq - \left(\sum_i \lambda_i \right)  \cdot \frac{1}{n^{8d}} - \norm{\cL}_{\infty} \E[|\gamma|] \\
& \geq -\frac{1}{n^{8d}} - n^{4d} n^{-8d}  > -\eta \mcom
\end{align*}
yielding a contradiction.
\end{proof}






\section{The Junta Approximation Theorem}
Here we prove \prettyref{thm:mainjunta} assuming \prettyref{thm:maindecomp}. In addition to the latter decomposition, the proof relies on certain \emph{extractor} properties of the inner-product function. Concretely, we need the following statement about the distribution of $G(X,Y)$ for blockwise-dense random variables $X,Y$ that is implicit in \cite{GLMWZ15}.

\begin{lemma}
Fix $q = 2^b$ for $b > 50$ and identify $[q]$ with $\on^b$.  Let $\iproduct : \zo^b \times \zo^b \to \on$ be the Boolean inner product function.

Suppose $X$ and $Y$ are independent, blockwise-dense random variables on $[q]^{n}$. Let $\nu$ be the density of the random variable $\iproduct^{\otimes n} (X,Y)$ on $\on^n$. Then, $\nu = 1 + h$ for an an $\epsilon$-decaying function $h$ where $\epsilon =  2^{-0.5b}.$ \label{lem:bd-epsilon-decay}
\end{lemma}

The lemma is an easy consequence of the fact that the inner-product function is a \emph{two-source extractor} for sufficiently high-entropies:

\begin{fact}[Chor-Goldreich \cite{CG88}]\label{fact:ipextractor}
Suppose $X, Y $ are independent random variables over $\on^{\ell}$ for $\ell > 7$ with min-entropy $H_{\infty}(X),H_{\infty}(Y) \geq 0.8 \ell$. Let $h:\zo^b \times \zo^b \rightarrow \on$ be the Boolean inner product function defined by $h(X,Y) = (-1)^{\oplus_{i = 1}^b X_i \cdot Y_i}.$ Then, 
\[|\E[h(X,Y)]| \leq 2^{-0.6 \ell + 1} \mper \]
\label{fact:Chor-Goldreich}
\end{fact}

Let $X', Y'$ be the random variables in $[q/2]^n$ obtained by removing the first bit in each block - i.e. by identifying $[q/2]$ with $\{-1,1\}^{b-1}.$ Then, since $X,Y$ are independent and blockwise-dense, for each $i$, $X'_{i}, Y'_{i}$ are independent random variables with min-entropy $H_{\infty}(X'_{i}), H_{\infty}(Y'_{i}) \geq 0.8 \ell-1.$ Applying Fact \ref{fact:Chor-Goldreich} to $X'_{i}$ and $Y'_{i}$, we obtain that $\|\E[h(X'_{i}, Y'_{i})]| \leq 2^{-0.6\ell+3} \mper$

By expanding out  $\E[\iproduct(X_{i},Y_{i})]$ by fixing each of the 4 values for the first bits in $X_{i},Y_{i}$ and using the upper bound above, we obtain that 
\begin{equation}
\label{eq:fourier-estimate-modified-ip}
\|\E[\iproduct(X_{i},Y_{i})]\| \leq 4 \cdot 2^{-0.6\ell+3} = 2^{0.6\ell + 5} \mper
\end{equation} 

\begin{proof}[Proof of Lemma \ref{lem:bd-epsilon-decay}]
By Fact \ref{fact:Chor-Goldreich}, for any $S \subseteq [n]$ we have
\[ \hat{\nu}(S) = \E[ \nu(z) \chi_S(z)] = \E_{z \sim \nu}[ \chi_S(z)] = \E[\Pi_{i \in I} \iproduct(X_{\{i\}},Y_{\{i\}})] = \E[\iproduct(X_{S},Y_{S})].\]
Here, $X_S = \prod_{i \in S} X_{\{i\}}$ and $Y_S = \prod_{i \in S} Y_{\{i\}}$.

Thus, using \eqref{eq:fourier-estimate-modified-ip}, $|\hat{\nu}(S)| \leq 2^{-0.6b|S|+5|S|} < 2^{-0.5b|S|}$ for $b > 50$.

Let $h(z) = \sum_{|S| \geq 1} \hat{\nu}(S) \chi_S(z).$ Then, $\nu = 1 + h$ and by the above estimate, $h$ is $\epsilon$-decaying for $\epsilon = 2^{-0.5b}.$
\end{proof}

Lemma \ref{lem:bd-epsilon-decay} showed that if $X,Y$ are $\bd$ random variables, then the density of $\iproduct^{\otimes n}(X,Y)$ is an $\epsilon$-decaying perturbation of the uniform density. In the following, we show a refinement of Lemma \ref{lem:bd-epsilon-decay} when $X,Y$ are aligned $d$-$\cbd$ random variables - specifically, that $g^{\otimes n}(X,Y)$ is a an $\epsilon$-decaying perturbation of a non-negative $d$-junta.

\begin{lemma}
Let $X,Y$ be aligned $d$-$\cbd$ random variables over $(\zo^{b})^n$ for $b > 7$ with the aligned blocks $I \subseteq [n]$ and $X_I = \alpha, Y_I = \beta$. Let $\nu$ be the density of $z = \iproduct^{\otimes n}(X,Y).$ Then, for $\epsilon  = 2^{-0.5b}$, there exists an $\epsilon$-decaying function $h$ such that
\[\nu = \1[z_I = \iproduct^{\otimes |I|}(\alpha, \beta)] \cdot (1+h).\] \label{lem:junta-approx-for-CBD}
In particular, if $u,v$ are aligned $d$-$\cbd$ densities over $(\on^b)^n$, then $Acc_{u,v}$ is a $(2^{-.5b},0)$-approximate conical $d$-junta.
\end{lemma}


\PRnote{need a consistent notation for gadget $g$,and its special case, $\iproduct$.  I have chosen one here, hopefully we can find something better}
\begin{proof}
Let $Z = \iproduct^{\otimes n}(X,Y)$. Then, $Z_I = \iproduct^{\otimes I}(\alpha,\beta$ and $Z_{\bar{I}} = \iproduct^{\otimes |\bar{I}|}(X_{\bar{I}}, Y_{\bar{I}})$. 
In particular, the density of $Z$ can be written as $\nu_I \cdot \nu_{\bar{I}}$ where $\nu_I$ is the density of $Z_I$ and $\nu_{\bar{I}}$ the density of $Z_{\bar{I}}$.

Now, by definition, $X_{\bar{I}}, Y_{\bar{I}}$ are $d$-$\cbd$ random variables. Thus, by \prettyref{lem:bd-epsilon-decay}, the density $\nu_{\bar{I}}$ of $Z_{\bar{I}}$ can be written as $1+h$ for a $2^{-0.5b}$-decaying function. This completes the proof of the first part of the statement. The next part follows from the definition of $Acc_{u,v}$.
\end{proof}

We are now ready to prove \prettyref{thm:mainjunta} which we restate for convenience.

\restatetheorem{thm:mainjunta}

%
\begin{proof}
We apply \prettyref{thm:maindecomp} to $u,v$ to write 
$$u \otimes v = \sum_{i =1}^N \lambda_i u_i \otimes v_i + \lambda_{err} \gamma_{err},$$
as guaranteed by the theorem. Then,
$$Acc_{u,v} = \sum_{i=1}^N \lambda_i Acc_{u_i,v_i} + \lambda_{err} \gamma(z),$$
where $\gamma$ denotes the distribution of $G(X,Y)$ for $(X,Y) \sim D_{error}$. Now, as $u_i,v_i$ are aligned $d$-\cbd densities, by \prettyref{lem:junta-approx-for-CBD} each $Acc_{u_i,v_i}$ is a $(2^{-.5b},0)$-approximate conical $d$-junta. Hence, $Acc_{u,v}$ is a $(2^{-.5b},\lambda_{err})$-approximate conical $d$-junta. The claim now follows from the bound on $\lambda_{err}$.
\end{proof}

\section{Decomposition of High Min-Entropy Distributions}
In this section we prove \prettyref{thm:maindecomp}. In fact, we show a stronger decomposition theorem that
is no more difficult to prove and is needed to recover the results of \cite{GLMWZ15} in our framework.
We will use the following notation:
\begin{itemize}
	\item We use distributions, densities and random variables interchangeably with the meaning being clear from the context.
	\item For $\mu$ a density on some domain $\cal{D}$ and $S \subseteq \cal{D}$, we write $\mu_{|S}$ for the density $\mu$ conditioned on $S$. We also define $\mu(S) = \pr_{X \sim \mu}[X \in S]$.
	\item For a random variable $X$ on $[q]^n$ and $I \subseteq [n]$, we write $X_I$ to denote $X$ projected to the coordinates in $I$. For a density $\mu$ on $[q]^n \times [q]^n$ and $I \subseteq [n]$, we write $\mu_I$ to denote the density on $[q]^{|I|} \times [q]^{|I|}$ obtained by projecting $\mu$ to the cooridnates in $I$. 
	\item For brevity, we say a density $\mu$ on $[q]^n \times [q]^n$ is an aligned $d$-\cbd if its two marginals along $[q]^n$ are aligned $d$-\cbd densities. 
\end{itemize}

\subsection{Warm Up: One-Dimensional Decompositions}

%

Observe that by definition, any $d$-\cbd density has min-entropy at least $0.8 \log q (n-d)$ which for $d \ll n$, we consider high. Thus, any convex combination of $d$-\cbd densities also has high min-entropy.
One could ask for a converse at this point: can every high min-entropy density be written as a convex combination of $d$-\cbd densities for small $d$? As a warmup for the more general decomposition, we first show that this is indeed the case. 


\begin{lemma} \label{lem:warmup-single-density}
Let $\mu$ is a density on $[q]^{n}$ with $H_\infty(\mu) \geq (n - t) \cdot \log q$.  Then, there exists a partition of $[q]^n$ as
\[[q]^n = \left(\bigcup_{i \in [N]} S_i\right) \cup S_{error} \]
such that
\begin{itemize}
\item For each $i \in [N]$, $\mu_{ \mid S_i}$ is a $10t$-\cbd distribution.
\item $\mu(S_{error}) \leq q^{-t}$.
\end{itemize}
\end{lemma}
\begin{proof}
We present an algorithm that obtains the claimed partition.

\begin{center}
  \fbox{\begin{minipage}{\textwidth}

  \textit{Setup:} $\mu$ is a distribution over $[q]^n$ with $H_{\infty}(\mu) \geq (n-t) \cdot \log q$. \\

  \textbf{Decompose(S)}\\
  \textit{Input} $S \subseteq [q]^n$. 
\begin{enumerate}
\item If $\mu_{|S}$ is $\bd$, \textsc{Terminate} and return $S$.
\item If $\mu(S) \leq q^{-t}$, \textsc{Terminate} and return $S$, labeled as $S_{error}$.
\item Else, let $Y \sim \mu_{|S}$.  Let $I \subseteq [n]$ be a \emph{maximal} set such that $H_\infty(Y_I) < 0.8 \log q  \cdot |I|$ and suppose $\alpha \in [q]^{I}$ be such that $p = \pr_{\mu_{|S}}[Y_I = \alpha] \geq q^{-0.8 |I|}$. Then, set $S_1 \lfta S \cap \{y \in [q]^n: y_I = \alpha\}$ and $S_2 \lfta S \cap \{y \in [q]^n: y_I \neq \alpha\}$.
\item \textsc{Return} $S_1 \cup  \textsc{Decompose}(S_2)$.
\end{enumerate}
  \end{minipage}}
\end{center}
To get the desired decomposition of $\mu$, we call $\textsc{Decompose}([q]^n)$.  Before we analyze the decomposition so produced, 
let us consider a single execution of the subroutine.  Consider an execution of Decompose$(S)$ that terminates in Step 4, returning a subset $S_1$ and calling Decompose$(S_2)$.  We make the following observations.

\begin{claim}
$\mu_{|S_1}$ is blockwise-dense except for the fixed coordinates in $I$.
\end{claim}
\begin{proof}
Since $I$ is a maximal set such that $H_{\infty}(Y_I) < 0.8 \log q |I|$, for any subset $J \subseteq [n]\slash I$, we have $H_{\infty}(Y_{J} | Y_{I} = \alpha) > 0.8 \log q |J|$.
\end{proof}
\begin{claim}
$|I| \leq 10t$
\end{claim}
\begin{proof}
For any $\beta \in [q]^{|I|}$, 
\begin{align*}
\Pr_{\mu_{|S}}[Y_{I} = \beta] 
& = \frac{\Pr_{\mu}[Y_{I} = \beta]}{\mu(S)} \\
& \leq \frac{\Pr_{\mu}[Y_{I} = \beta]}{q^{-t}} \qquad \qquad (\text{Step(2) did not terminate}) \\
& \leq \frac{q^{t-|I|}}{q^{-t}} = q^{2t - |I|} \qquad \qquad (H_{\infty}(\mu) \geq (n-t)(\log q)) \\
\end{align*}
Since there exists $\alpha \in [q]^{I}$ such that $\Pr_{\mu_{|S}}[Y_{I} = \alpha] \geq q^{-0.8 |I|}$, we get that $q^{-0.8 |I|} \leq q^{2t - |I|}$.  This implies that $|I| \leq 10t$. 
\end{proof}
From the above claims, it is clear that $\mu_{|S_1}$ is a $10t$-CBD distribution, whenever Decompose$(S)$ terminates in Step (4).

Suppose we call Decompose$([q]^n)$, the recursive algorithm will return a partition of $[q]^n$ into subsets $\{S_i\}_{i \in [N]}$ and eventually terminate either via Step(1) or Step(2).
If the algorithm terminates via Step (1), then $\mu_{|S_i}$ is $10t$-CBD for all the sets $S_i$ and the lemma follows.
If the algorithm terminates via Step (2), then it produces a subset $S_{error}$ with $\mu(S_{error}) \leq q^{-t}$, as desired.
\end{proof}
\subsection{Rectangular Decompositions}
To prove our junta theorem, \prettyref{thm:mainjunta}, and for other plausible applications in communication complexity, the decomposition obtained by  \prettyref{lem:warmup-single-density} does not suffice.  In particular, the underlying domain is two-dimensional $[q]^n \times [q]^n$, and the partitions need to be {\it rectangular}.  In this section, we will prove a general {\it rectangular decomposition} theorem designed for distributions with high min-entropy.


A combinatorial rectangle $\cR \subseteq [q]^n \times [q]^n$ is given by $\cR =  \cA \times \cB$ for $\cA,\cB \subseteq [q]^n$. We can now state our main decomposition theorem. \ignore{For a density $\mu$, we will denote by $\mu_{|R}$ the density $\psi$ conditioned on being within $R$, i.e.,
\[\mu_{|R}(x,y) = \mu(x,y) \cdot \frac{\Ind[ x \in \cA \wedge y \in \cB ] }{\mu(R)} \mcom\]
where $\mu(R)$ is total probability mass of $\mu$ within $R$, namely $\mu(R) = \Pr_{\mu}[ x \in \cA \wedge y \in \cB]$.}

\begin{theorem}[Rectangular Decompositions]\label{thm:maindecompfull}

Let $\mu$ be a probability density on $[q]^n \times [q]^n$ such that for all $I \subseteq [n]$,
\[ H_{\infty}(\mu_{I}) \geq 1.9 \cdot \log{q} \cdot |I| - t \mper\]
Then for all $d \in \N$, there exists a partition
\[ [q]^n \times [q]^n = \left(\bigcup_{i \in [N]} \cR_i \right)\cup \Error\] where $\{\cR_i = \cA_i \times
\cB_i\}_{i \in [N]}$ are rectangles such that
\begin{enumerate}
\item For each $i \in [N]$, $(X_i,Y_i) \sim \mu_{\mid \cR_i}$, $X_i,Y_i$ are aligned $d$-$\cbd.$
\item $\mu(\Error) \leq 2^t \cdot (dq^{-0.05})^d$
\end{enumerate}
\label{thm:alg-analysis}
\end{theorem}

Notice that the above theorem also implies that the density $\mu$
can be approximated by a convex combination of aligned $d$-\cbd distributions
by setting,
\[ \mu = \sum_{i \in [N]} \mu(\cR_i) \cdot \mu_{|\cR_i} + \mu(\Error)
\cdot \mu_{|\Error} \mper\]
\prettyref{thm:maindecomp} is an an immediate consequence of the above decomposition theorem.  
\begin{proof}[Proof of \prettyref{thm:maindecomp}]
Let $\mu$ be the density $u \otimes v$ on $[q]^n \times [q]^n$. Then, $\mu$ for all $I \subseteq [n]$, $H_\infty(\mu_I) \geq 2 (\log q) |I| - t$.  Now, apply \prettyref{thm:maindecompfull} to obtain a partition of $[q]^n \times [q]^n$ as $\cup_{i \in N} \cR_i \cup \Error$ satisfying the conditions of the theorem. Then,
$$\mu = \sum_{i=1}^N \mu(\cR_i) \cdot \mu_{|\cR_i} + \mu(\Error) \mu_{|\Error}.$$
Note that $\sum_{i=1}^N \mu(\cR_i) + \mu(\Error) = 1$.
Let $\cR_i = \cA_i \times \cB_i$. Then, we can write $\mu_{|\cR_i} = u_i \otimes v_i$, where $u_i = u_{|\cA_i}$ and $v_i = v_{|\cB_i}$. Then, $X_i \sim u_i, Y_i \sim v_i$ are aligned $d$-\cbd and 
$$\mu(\Error) \leq 2^{t} (dq^{-0.05})^d \leq 2^{-t} (nq^{-0.05})^d \leq 2^{-t} q^{-\Omega(d)} \leq q^{-\Omega(d)},$$
for $q \geq n^c, d \geq c t$ for a sufficiently big constant $c \geq 1$. This proves the theorem.
\end{proof}
%

\subsection{The Decomposition Algorithm}\label{sec:decompalgorithm}

We will show \prettyref{thm:alg-analysis} by devising an algorithm
that constructs the partition given the distribution $\mu$ and a
parameter $d \in \N$. The algorithm is an natural extension of the one used in the proof of \prettyref{lem:warmup-single-density} and is similar to the one used in \cite{GLMWZ15}; however, our analysis is quite different from theirs. Indeed, while they also obtain a similar decomposition theorem, the error guarantee is not exponentially small as we obtain and is needed in our application.  

The formal description of the algorithm is at the end of this
subsection. For exposition, we depict a labeled execution tree of the
algorithm \prettyref{fig:exec-tree}. \Dec is the main procedure that takes as input a rectangle $\cA \times
\cB \subseteq [q]^{n} \times [q]^{n}$. This rectangle will always satisfy the invariant of
having an \emph{aligned} set of nodes \emph{fixed} - i.e. there is an
explicitly identified set of indices $F \subseteq [n]$ such that for some two fixed strings $\alpha, \beta$, for all $(x,y) \in \cA \times \cB$, $x_F = \alpha$ and $y_F = \beta$. Observe that in the beginning, $\cA = \cB = [q]^{n}$ and $F = \emptyset.$

Each time $\Dec$ is invoked by the algorithm, we create a new node in the execution tree and identify it as being created by $\Dec$ by indexing it with $v_1, v_2, \ldots, $. If the set of fixed blocks $F$ in the input rectangle $\cR$ has size $\geq d$, the algorithm terminates and adds the associated rectangle $\cR$ to $\Errorb$; $\Errorb$ contains the set of rectangles that account for error owing to the number of fixed blocks in them exceeding $d$.  Next, if $\mu(\cR) < \delta$, then $\cR$ is added to $\Errora$; $\Errora$ maintains the collection of rectangles that are labeled as error because their measure was too small.

Now, suppose that the input rectangle $\cR$ does not satisfy the conditions of $\Errora$ or $\Errorb$. If $\mu_{\mid \cR}$ is $\bd$, then, we terminate the algorithm and return $\cR$. Otherwise, there exists $S \subseteq [n]$ and some assignment to variables in $S$, say $\alpha_S$ such that $\Pr_{\mu_{\mid \cR}}[ X_S = \alpha_S ] > q^{-0.8|S|}$ (or $\Pr_{\mu_{\mid \cR}}[ Y_S = \alpha_S] > q^{-0.8|S|}).$  The idea is to split the rectangle $\cA \times \cB$ into two rectangles, $\cA_{|S = \alpha_S} \times \cB$ and $\cA_{|S \neq \alpha_S} \times \cB$; here, we define $\cA_{|S = \alpha_S}$ denotes the set,
\[ \cA_{|S = \alpha_S} = \cA \cap \{x: x_{S} = \alpha_S\} \text{ and } \cA_{|S \neq \alpha_S} = \cA \cap \{x: x_{S} \neq \alpha_S\}\mper  \]
In the rectangle $\cA_{|S = \alpha_S} \times \cB$, $X$ and $Y$ don't have the same set of fixed blocks, since $X$ is fixed in $F \cup S$ while $Y$ is fixed only on $F$.  
To remedy this, the subroutine $\XDec$ (or $\YDec$, respectively) is executed on the rectangle $\cA_{|S = \alpha_S} \times \cB$.  The \Dec routine continues with the remaining rectangle $\cA_{\mid S \neq \alpha_S} \times \cB$.  Each call to $\XDec$ or $\YDec$ is denoted by a node in the execution tree labeled by $w_1, w_2, \ldots.$ 

The subroutine $\XDec$ (the case of $\YDec$ is analogous) takes as input the rectangle $\cA \times \cB$ along with the fixed set $F$ that was the current input of the $\Dec$ routine when $\XDec$ was invoked in addition to the new set of indices $S$ that violated blockwise-density. $\XDec$ then chooses every possible value $\beta$ for $Y_S$ and for each $\beta$, calls $\Dec$ recursively with the rectangle $\cA \times \cB_{\mid S = \beta}$ with $F \cup S$ as the set of fixed coordinates.

\begin{center}
  \fbox{\begin{minipage}{\textwidth}
  \captionof*{figure}{Decomposition Algorithm}
  \label{alg:dec}
  \textit{Setup:} \begin{itemize}\itemsep=0ex
		  \item           A probability density $\mu$ such that  $ H_{\infty}(\mu_{\mid I}) \geq 1.9 \log{q} \cdot |I| - t$ for all $I \subseteq [n]$. 
     \item A parameter $d \in \N$ and let $\delta \defeq q^{-0.05d}$. Set $\Errora = \Errorb = \emptyset$.
                  \end{itemize}
\textbf{Decompose$(\cA \times \cB,F)$}\\
\textit{Input:}   A rectangle $\cA \times \cB \subseteq [q]^n \times [q]^n$, $F \subseteq [n]$: subset of ``fixed'' indices.  \\
  \textit{Invariant:} $\cA_F,\cB_F$ are fixed.
\begin{enumerate}
\item If $|F| \geq d$, \textsc{Terminate} after setting $\Errorb \lfta \Errorb \cup \cA \times \cB$.
\item Set $\cR \equiv \cR_0 \equiv \cA \times \cB$.
\item While $\mu(\cR) \geq \delta \cdot \mu(\cR_0)$ do
	\begin{enumerate}
		\item Let $(X,Y) \sim \mu_{|\cR}$. If $X_{\overline{F}}, Y_{\overline{F}}$ are $\bd$, \textsc{Terminate}.
        \item Else, if there is an $S \subseteq [n] \backslash F$ and $\alpha \in [q]^S$ such that
				$$ \Pr[X_S = \alpha] > q^{-0.8  |S|},$$ call $\YDec$ on input $(\cA_{|S=\alpha},\cB,F,S)$ and set $\cR \leftarrow \cA_{|S \neq \alpha} \times \cB$

	 	\item Else, if $\Pr[Y_S = \alpha] > q^{-0.8  |S|}$ then call $\XDec$ on input $(\cA,\cB_{|S = \alpha},F,S)$ and set $R \leftarrow \cA \times \cB_{|S \neq \alpha}$.
	\end{enumerate}
\item Set $\Errora \leftarrow \Errora \cup R$.
\end{enumerate}
  \end{minipage}}

  \fbox{\begin{minipage}{\textwidth}
  \textbf{XDecompose$(\cA \times \cB, F, S)$} \\
 \textit{Input:} A rectangle $\cA \times \cB$, $F \subseteq [n]$: common subset of
 ``fixed'' indices; $S$: the set of coordinates newly fixed in $\cA$ (but not in $\cB$).\\
  \textit{Invariant:} $X_{F}$ and $Y_{F \cup S}$ are fixed.
  \begin{enumerate}
  	\item For every $\beta \in [q]^S$,
		\textsc{Decompose}$(\cA_{|S = \beta},\cB,F\cup S)$.
  \end{enumerate}

  \end{minipage}}

  \fbox{\begin{minipage}{\textwidth}
  \textbf{YDecompose$(\cA \times \cB, F,S)$} \\
   \textit{Input:} A rectangle $\cA \times \cB$, $F \subseteq [n]$: common subset of ``fixed'' indices. \\
  $S$: the set of coordinates newly fixed in $\cB$ (but not in $\cA$). \\
  \textit{Invariant:} $X_{F \cup S}$ and $Y_F$ are fixed.
  \begin{enumerate}
  	\item For every $\beta \in [q]^S$,
		\textsc{Decompose}$(\cA,\cB_{|S=\beta},F\cup S)$.
  \end{enumerate}

  \end{minipage}}

\end{center}

\begin{figure}[h]
\begin{center}
\includegraphics[scale=0.50]{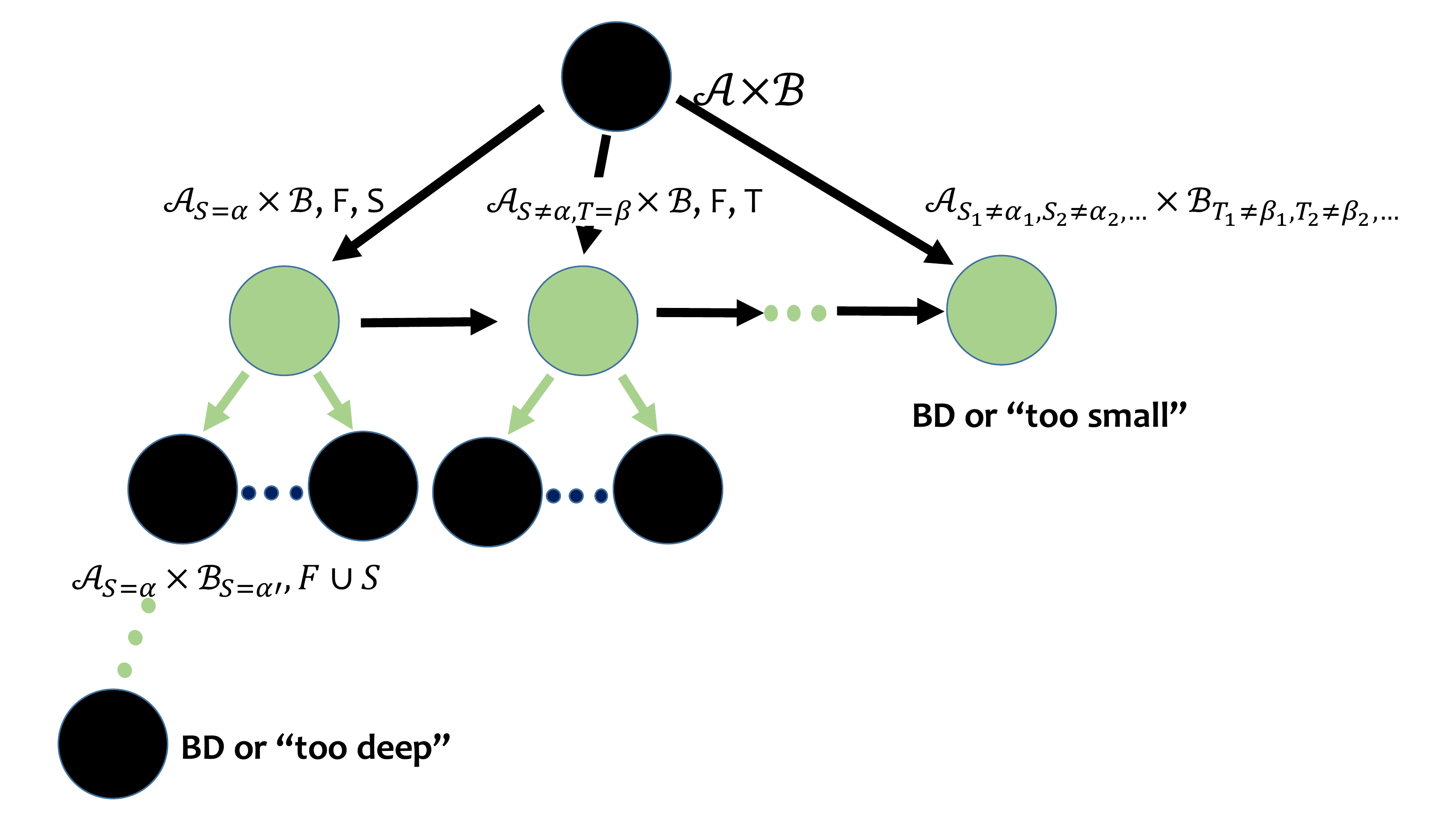}
\caption{Execution Tree: Black Nodes are Calls of $\Dec$, Green Nodes are calls of $\XDec$ or $\YDec$. ``too deep"  $\equiv \Errora$, ``too small'' $\equiv \Errorb$ and BD = $\bd$}\label{fig:exec-tree}
\end{center} 
\end{figure}


\subsection{Analysis: Proof of \prettyref{thm:alg-analysis}}
We now analyse the algorithm. 

\paragraph{Execution Tree} Consider the execution tree of the decomposition algorithm  (Figure \ref{fig:exec-tree}). We associate a node $\rho$ corresponding to a call of $\Dec$ with parameters $\cA_\rho, \cB_\rho$ and $F_\rho$ for $\cA_\rho, \cB_\rho \subseteq [q]^n $ and $F_\rho \subseteq [n].$ Nodes $\rho$ corresponding to a call of $\XDec$ or $\YDec$ are associated with an additional parameter $S_\rho \subseteq [n]$. The calls of $\Dec$ and $\XDec$ or $\YDec$ alternate.

Before we analyze the decomposition, we introduce some notation.
\newcommand{\depth}{\mathsf{depth}}
\begin{itemize}
\item For each vertex $\rho$, let $\mu(\rho) = \mu(\cA_{\rho} \times \cB_\rho)$.  Let $R(\rho)$ denote the rectangle $\cA_{\rho} \times \cB_{\rho}$.
\item For a child $b$ of a vertex $a$, let $\mu(b|a) = \mu(b)/\mu(a)$.

\item We will reserve the letter $v$ (and various suffixes) for nodes corresponding to calls of $\Dec$ (i.e., the vertices in odd-layers) and the letter $w$ for nodes corresponding to calls of $\XDec$ or $\YDec$ (i.e., the vertices in even-layers).

\item For a vertex $v$ let $C_v = (w_1,\ldots,w_{c_v})$ be the children of $v$ (in the order they were generated), where we assume that $w_1,\ldots,w_{c_v-1}$ lead to recursive calls of $\XDec$ or $\YDec$ while $w_{c_v}$ corresponds to the rectangle marked as $\Errora$ in Step(4) of $\Dec$.

\item For a vertex $v$, and $w_i \in C(v)$, let $S_{w_i} \subseteq [n]$ denote the corresponding set of new blocks that were fixed to produce $w_i$.

\item We say a node $v$ in the tree is {\it bad} if $|F_v| \geq 2d$ but $\mu_{|\cA_{v}\times \cB_{v}}$ is not $\bd$ in the non-fixed blocks.  A root-to-leaf path $v_1,w_1,v_2,w_2,\ldots,
	v_t, w_t,v_{t+1}$ is called a {\it bad path} if $v_{t+1}$ is bad.
\end{itemize}

Note that the tree yields a partition of $[q]^n \times [q]^n$ as,
\[ [q]^n \times [q]^n = \left(\bigcup_{v \in \cL]} \cR(v) \right)\cup \Errora \cup \Errorb,\]
where $\cL$ denotes the leaves of the tree not added to $\Errora$ or $\Errorb$. We will show that this partition satisfies the conditions of the theorem when we take $\Error = \Errora \cup \Errorb$. Note that, all the leaves in the execution tree correspond to the calls of $\Dec$. 

We first argue that $\mu_{|\cR(v)}$ for $v \in \cL$ gives $d$-$\cbd$ densities in the non-fixed coordinates. 

\begin{lemma}[Good Rectangles]
Let $v \in \cL$. Then, $\mu_{\mid \cR(v)}$ is an aligned $d$-$\cbd.$ \label{lem:proof-part-1}
\end{lemma}
\begin{proof}
The assumptions imply that $|F_v| \leq d$. Further, as $\cR(v)$ was not added to $\Errora$, $\cR(v)$ is a leaf because for $(X,Y) \sim \mu_{\mid \cR(v)}$, $X_{\bar{F}}, Y_{\bar{F}}$ are $\bd$. Thus, $\mu_{\mid \cR(v)}$ is an aligned $d$-$\cbd$ density. 

\end{proof}

It is also easy to bound $\mu(\Errora)$:

\begin{lemma}
$\mu(\Errora) \leq d \delta.$ \label{lem:proof-part-2}
\end{lemma}

\begin{proof}
Note that each call of \textsc{Decompose} leads to at most one rectangle being added to $\Errora$. Further, if $\textsc{Decompse}(R_0)$ led to adding a rectangle $\cR \subseteq \cR_0$ to $\Errora$, then $\mu(R) < \delta \cdot \mu(R_0)$. Now, in any single layer of the execution tree the nodes corresponding to calls of $\Dec$ are associated with disjoint rectangles. Thus, the total measure of all the rectangles that are included in $\Errora$ due to $\Dec$ calls from nodes in a specific layer is at most $\delta$. As there are at most $d$ layers that have $\Dec$ nodes, it follows that $\mu(\Errora) \leq d \delta$.
\end{proof}

The main task is to prove an upper bound on $\mu(\Errorb)$.
\paragraph{Bounding $\mu(\Errorb)$}

We begin with an important definition. 

\begin{definition}[$\theta(w_i \mid v)$]
For a node $v$ associated with a $\Dec$ call, let $C_v = (w_1,\ldots,w_{c_v})$ be the children of $v$. Define 
$$\theta(w_i \mid v) = \sum_{j=i}^{c_v} \mu(w_j \mid v).$$ 
\end{definition}
Intuitively, $\theta(w_i \mid v)$ denotes the relative measure of the rectangle $\cR$ inside Step $4$ of Algorithm \ref{alg:dec} just before the call of $\XDec$ or $\YDec$ associated with the node $w_i$. 

\ignore{Let $w_i$ be the child associated with the $i^{th}$ call of $\XDec$ or $\YDec$ inside the call of $\Dec$ associated with some node $v$.

\begin{definition}[$\theta(w_i \mid v)$]
Let $\theta(w_i \mid v)$ denote the relative measure under $\mu$ of the rectangle $R$ inside Step $4$ of Algorithm \ref{alg:dec} just before the call of $\XDec$ or $\YDec$ associated with the node $w_i$.
\end{definition}

$\theta(w_i \mid v)$ can be expressed in terms of $\mu$ as follows:

\begin{lemma}[Measure at $w$ Nodes]
Let $w_1, w_2, \ldots,$ be the $w$-nodes that are children of $v$ in the execution of the algorithm. Then,
$$\theta(w_i|v) = \sum_{j=i}^{c_v} \mu(w_j|v).$$\label{lem:def-theta}
\end{lemma}
\begin{proof}
Let $R_i$ be the rectangle that step 4 of $\Dec$ works with just before the call of $\XDec$ or $\YDec$ that creates the node $w_i.$ Each iteration of the while loop in step 4 of the algorithm further partitions $R$. The total $\mu$-measure of $R_i$ can thus be obtained by summing up the measures of all the rectangles $R_j$ that occur in $\Dec$ at $v$ at or after the call of $\XDec$ or $\YDec$ that creates $w_i.$
\end{proof}}


\begin{lemma}\label{lem:density}
Fix any vertex $v$ and a child $w$ of $v$,
$$\mu(w|v) \geq q^{-0.8 |S_w|} \cdot \theta(w|v).$$
\end{lemma}
\begin{proof}
Let $\cR$ be the rectangle processed in the while loop inside the call of $\Dec$ associated with $v$ just before the call to create $w$. Suppose that $w$ was created by a recursive call to $\XDec$ for $S_w \subseteq [n]$ such that $\Pr_{(X,Y) \sim \mu_{|\cR}}[ X_{S_w} = \alpha] \geq q^{-0.8 |S_w|}$ (the case of a call to $\YDec$ can be dealt with analogously). Thus, $\mu_{|\cR}(\cR(w)) \geq q^{-0.8 |S_w|}$.

Now, we can write  $\mu(w | v) $ as a product of the relative probability of $\cR$ in $\cR(v)$ and the relative probability of $\cR(w)$ under $\cR$. Concretely, 
$$\mu(w | v) = \mu_{|\cR(v)}(\cR(w)) = \mu_{|\cR(v)}(\cR) \cdot \mu_{|\cR}(\cR(w)) = \theta(w|v) \cdot \mu_{|\cR}(\cR(w)) \geq \theta(w |v) \cdot q^{-0.8 |S_w|},$$
where the last equality follows from the definition of $\theta(w|v)$. 


\end{proof}

Next, we estimate $\mu(v)$ for every $v$ that is a \emph{bad} leaf i.e., a leaf at depth $2d+1.$  We begin by showing a bound on $\mu(v)$ for an arbitrary node $v$.
\begin{lemma} \label{lem:node-measure-bound}
For any vertex $v$ in the execution tree,
\[ \mu(v) \leq 2^t \cdot q^{- 1.9 |F_v|} \mper\]
\end{lemma}
\begin{proof}
	The rectangle $\cR(v)$ corresponding to $v$ has the blocks $F_v$ fixed, while the values in the remaining blocks could also be constrained.  Let us suppose $X_{F_v} = \alpha$ and $Y_{F_{v}} = \beta$ for $(X,Y) \in \cR(v)$.  By our assumption on $\mu$,
\[ H_{\infty}(\mu_{|F_v}) \geq 1.9 \cdot \log q \cdot |F_v| - t \mcom\]
which implies that,
\[ \mu(v) \leq \Pr_{\mu} [X_{F_v} = \alpha \wedge Y_{F_v} = \beta ] \leq 2^{-H_{\infty}(\mu_{F_v})} \leq 2^t \cdot q^{-1.9 |F_v|} \mper\]
\end{proof}

\begin{lemma}\label{lem:badleafbound}
Let $v_1,w_1,\ldots, v_\ell,w_\ell,v_{\ell+1}$ be a {\emph bad} path ending at a node that is added to $\Errorb$. Then,
	$$\mu(v_{\ell+1}) \leq q^{-1.1 |F_{v_{\ell+1}}|} \cdot 2^t  \cdot
	\prod_{i=1}^{\ell} \frac{\mu(w_i|v_i)}{\theta(w_i | v_i)}
	$$ \label{lem:bad-path-prob}
\end{lemma}

\begin{proof}
Let $s = |F_{v_{\ell+1}}|$; then, by \prettyref{lem:node-measure-bound},
	\begin{align*}
    \mu(v_{d+1}) & \leq q^{-1.9 s} \cdot 2^t
    \end{align*}

On the other hand, we know that $\sum_{i=1}^\ell |S_{w_i}| = s$. Thus, by \lref{lem:density},
    $$
	\mu(v_{\ell+1}) = \prod_{i = 1}^\ell \mu(v_{i+1}|w_i) \mu(w_i |
	v_i)
	\geq \prod_{i = 1}^\ell \mu(v_{i+1}|w_i) \cdot \theta(w_i |
	v_i) \cdot q^{-0.8  |S_{w_i}|}.
$$ The above two
	inequalities imply that,
	$$ \prod_{i=1}^\ell \mu(v_{i+1}|w_i) \theta(w_i | v_i) \leq
	q^{-1.1 s} \cdot 2^t.$$

	The claim now follows from the above
	inequality along with
	$$\mu(v_{\ell+1}) = \prod_{i=1}^{\ell}
	\mu(v_{i+1}| w_{i}) \mu(w_i | v_i)  = \bkets{\prod_{i=1}^\ell \mu(v_{i+1}| w_i) \theta(w_i | v_i)} \cdot \prod_{i=1}^\ell \frac{\mu(w_i | v_i)}{\theta(w_i | v_i)}.$$

\end{proof}

We need the following elementary lemma.
\begin{lemma}
	Let $a_1,\ldots, a_N \in (0,1) $ be such that $\sum_{i = 1}^{N} a_i = 1$ and $a_{N-1} + a_N \geq \eps$. 
	Then,
	$$ \sum_{j = 1}^{N} \frac{a_j}{ \sum_{i \geq j} a_i}  \leq
	\lceil \log
	(1/\eps) \rceil + 2.$$
\end{lemma}

\begin{proof}
For $i < N$, let $s_i = \sum_{j \geq i} a_j$;  clearly, $s_i$ is a
decreasing sequence and $s_{N-1} > \epsilon$.
Let $\ell_i$ be the largest index such that
$s_{\ell_i} = \sum_{j = \ell_i}^{N} a_j \geq \frac{1}{2^i}$.

By definition,
$$\sum_{j = \ell_i+1}^{\ell_{i+1}} a_j \leq \sum_{j = \ell_i+1}^{N} a_j
< \frac{1}{2^{i}} $$

Let $t = \lceil \log(1/\eps) \rceil$.
Clearly, $\ell_{t} \geq N-1$.
Now, we have,
\begin{align*}
	\sum_{i = 1}^{N-1} \frac{a_i}{s_i} &= \sum_{i = 0}^{t}
	\sum_{j = \ell_i + 1}^{\ell_{i+1}} \frac{a_j}{s_j} \\
	& \leq 	\sum_{i = 0}^{t}
	\frac{ \sum_{j = \ell_i + 1}^{\ell_{i+1}} a_j}{s_{\ell_{i+1}}} \\
	& \leq \sum_{i=0}^{t} 1  = t+1.
\end{align*}

\end{proof}

We have the following immediate consequence of the above lemma.

\begin{corollary} \label{cor:sumnode}
For every vertex $v$,
$$ \sum_{w \in C_v} \frac{\mu(w|v)}{\theta(w|v)} \leq
\lceil \log(1/\delta) \rceil +2. $$
\end{corollary}
\begin{proof}
Let $C_v = (w_1,\ldots,w_{c_v})$. Then, $\theta(w_{c_{v-1}} | v) = \mu(w_{c_{v-1}} | v ) + \mu(w_{c_v} | v) \geq \delta$. As $\theta(w_i | v) = \sum_{j \geq i} \mu(w_j | v)$, the claim now follows by applying the previous lemma to the numbers $\mu(w_1|v),\mu(w_2 | v),\ldots,\mu(w_{c_v} | v)$. 
\end{proof}

\begin{lemma}
	The sum over all leaves
	$$ \sum_{\text{ paths } v_1,w_1,\ldots, w_{\ell-1},v_\ell}
	q^{-|F_{v_\ell}|} \prod_{i = 1}^{\ell}
	\frac{\mu(w_i|v_i)}{\theta(w_i|v_i)} \leq \bkets{\lceil \log 1/\delta\rceil + 2}^{d}$$ \label{lem:path-estimate}
\end{lemma}
\begin{proof}
	For each $v$ corresponding to a \Dec call, define a probability distribution $\gamma(\;|v)$ over
	$C_v$ as follows: for each $w \in C_v$, $$\gamma(w|v) = \frac{\mu(w|v)}{\theta(w|v)} \cdot
	\left(\sum_{w \in C_v} \frac{\mu(w|v)}{\theta(w|v)} \right)^{-1}.$$ 
	
	For each $w$ corresponding to a $\XDec$ or $\YDec$ call consider the uniform distribution over the children of $w$. Then, these distributions induce a probability distribution over the leaves of the execution tree: the probability of any leaf $v_{\ell+1}$ is given by $\prod_{i = 1}^{\ell} q^{-|S_{w_i}|} \cdot
	\gamma(w_i|v_i)$, where $v_1,w_1,\ldots, v_\ell ,w_\ell, v_{\ell+1}$ is the path from the root to $v_{\ell+1}.$

By our construction, the total probability under the above distribution of all the leaves is $1.$ Thus, the sum over all leaves
$$ \sum_{\text{ paths } v_1,w_1,\ldots, v_\ell w_\ell}
	\prod_{i = 1}^{\ell} q^{-|S_{w_i}|} \cdot
	\gamma(w_i|v_i) = 1. $$

    By Corollary \ref{cor:sumnode}, $\gamma(w_i | v_i) \geq
	\frac{\mu(w_i|v_i)}{\theta(w_i|v_i)} \cdot
	\frac{1}{\bkets{\lceil \log 1/\delta\rceil + 2}} $.  The result follows from substituting this into the previous expression and using the fact that $\ell \leq d$ (number of fixed coordinates is at most $d$).
\end{proof}

We are now ready to bound $\mu(\Errorb)$.

\begin{lemma} $\mu(\Errorb) \leq q^{-0.1 d} \cdot 2^t  \cdot
	\bkets{\lceil \log 1/\delta\rceil + 2}^{d}$.\label{lem:proof-part-3}
\ignore{
	Suppose the distribution $\mu$ satisfies $H_{\infty}(\mu_{I}) \geq 1.9 \cdot \log q \cdot |I| - t$ then
		$$ \sum_{v_{d+1} \in Bad} \mu(v_{d+1}) \leq q^{-0.1 d} \cdot 2^t  \cdot
	\bkets{\lceil \log 1/\delta\rceil + 2}^{d}.
	$$ \label{lem:proof-part-3}}
\end{lemma}
\begin{proof}
Let \emph{Bad} denote all leaves that resulted in rectangles being added to $\Errorb$. The proof follows by using Lemmas \ref{lem:bad-path-prob} and Lemma \ref{lem:path-estimate}. For brevity, in the following, let $v_1,w_1,\ldots,v_\ell,w_\ell,v_{\ell+1} = v$ be the path to $v$ from the root. 
	\begin{align*}
		\sum_{v \in Bad} \mu(v) & \leq \sum_{v \in Bad} q^{-1.1 |F_{v}|} \cdot 2^t  \cdot \prod_{i=1}^{\ell} \frac{\mu(w_i|v_i)}{\theta(w_i | v_i)} \\
	& \leq q^{-0.1d} \cdot 2^{t} \left(\sum_{v \in Bad} q^{- |F_{v}|}  \cdot \prod_{i=1}^{\ell} \frac{\mu(w_i|v_i)}{\theta(w_i | v_i)} \right)
	\leq q^{-0.1d} \cdot 2^{t} \cdot \bkets{\lceil \log 1/\delta\rceil + 2}^{d}.
\end{align*}
\end{proof}

\begin{proof}[Proof of \prettyref{thm:alg-analysis}]
One executing \prettyref{alg:dec}, we obtain a partition of $[q]^n \times [q]^n$ into $\cup_{v \in \cL} \cR(v) \cup \Error$, where we define $\Error = \Errora \cup \Errorb$. 

	By \prettyref{lem:proof-part-1}, $\mu_{|\cR(v)}$ is an aligned $d$-\cbd for each $v \in \cL$.  Furthermore, the total measure of $\Error$ is 
\begin{align*}
	 \mu(\Error) & = \mu(\Errora) + \mu(\Errorb) \\
	 & \leq d\delta + q^{-0.1 d} \cdot 2^t  \cdot	\bkets{\lceil \log 1/\delta\rceil + 2}^{d} \mcom \qquad (\text{\prettyref{lem:proof-part-2} and \prettyref{lem:proof-part-3}})\\
	 & \leq  2^t \cdot (d q^{-0.05 })^d \qquad \qquad \qquad (\text{ for } \delta = q^{-0.05d}).
 \end{align*}
 This completes the proof.
\end{proof}

\section{Non-negative Rank and LP Lower Bounds for CSPs}
\label{sec:nnrlift}
In this section, we show that proving a lower bound on the size of linear programming relaxations for CSPs reduces to proving non-negative rank lower bound on a pattern matrix $M_f$ for an appropriate choice of $f$. We then use this characterization to prove \prettyref{thm:maincsp}, \prettyref{cor:maincsp}. We first show the following.

\begin{lemma}[LP Lower Bounds from Pattern Matrices] \label{lem:lplb-via-pattern}
Let $0 < s < c \leq 1$ and let $\Lambda:\on^k \to \zo$ be a predicate and let $\I^*$ be an instance of $CSP(\Lambda)$ on $n$ variables such that $\opt(\I^*) \leq s$ and let $f(x) = c- \I^*(x)$.
For all $g: [q] \times [q] \to \on$,  any linear programming relaxation of CSP(P) that achieves $(c,s)$-approximation on instances with $q \cdot n$ variables, has size at least $R \geq \Omega(\nnr(M^g_f)).$
\end{lemma}

The above lemma is an easy consequence of the characterization of size of linear programs  for CSPs \cite{CLRS13,Yan88}.
\PRnote{it is easier to state a lower bound result for lp size, than a characterization.  for lower bound, it is sufficient to define the matrix for plantings of an instance, whereas for characterization, one needs to define a single matrix for all instances.  I am going to state the characterization, may be we can scale back later}
Towards stating the characterization, let us fix a CSP $\Lambda$.  For $n \in \N$ and $s \in [0,1]$, let $\Lambda^s_n$ denote the family of all instances of the CSP on $n$ variables with $\opt(\cI) \leq s$.  With these definitions, we are ready to state the characterization.
\begin{fact} \cite{CLRS13} \label{fact:lpsize}
The size of the smallest linear program that $(c,s)$-approximates a CSP $\Lambda$ on instances with $n$ variables is $\Theta(\nnr(\cM_{n,s}))$ where the matrix $\cM_{n,s} : \Lambda^s_n \times \on^n \to \R$ is defined as
\[  \cM_{n,s}(\cI,x) \defeq c- \cI(x) \mper\]
\end{fact}
\PRnote{no reference for the above fact :(}

To prove \prettyref{lem:lplb-via-pattern}, we will show that $M_f^g$ is a sub-matrix of $\cM_{q \cdot n,s}$, and therefore $\nnr(\cM_{q \cdot n,s}) \geq \nnr(M_f)$.  By \prettyref{fact:lpsize}, this implies that the size of any LP that $(c,s)$-approximates the CSP on instances with $q \cdot n$ is at least $\nnr(M_f)$.  Hence, \prettyref{lem:lplb-via-pattern} is an immediate consequence of the following.

\begin{lemma} \label{lem:submatrix}
Let $\I^{*}$ be an instance of CSP on $n$ variables such that $\opt(\I^*) = s$ and let $f(x) = c- \I^*(x)$.  Let $g: [q] \times [q] \to \on$ be a gadget function. Then, $M_f^{g}$ is a sub-matrix of $\cM_{q \cdot n, s}$.
\end{lemma}
\begin{proof}
%
For $\alpha \in [q]$, define its \emph{$g$-encoding} $g \circ \alpha \in \on^{q}$ consisting of the truth-table of the function $g \circ \alpha : \beta \to g(\alpha,\beta)$.  Specifically, if we index the coordinates of $g \circ \alpha$ with $\beta \in [q]$, then $(g \circ \alpha)_{\beta} \defeq g(\alpha,\beta)$.  Similarly for $x \in [q]^{n}$, define $\tilde{x} \in \on^{q \cdot n}$ by setting $ \tilde{x}_i \seteq g \circ x_i$.

For each $y \in [q]^n$, we create an instance $\cI_y$ on $N = q \cdot n$ variables.  We will index the variables of $\cI_y$ by $[n] \times [q]$ and denote them by $\{z_{i,\beta} | i \in [n], \beta \in [q] \}$.  The instance $\cI_y$ is obtained by planting the instance $\I^*$ on the subset of variables $\{z_{1,y_1},\ldots,z_{1,y_n}\}$.

By definition of the matrices $\cM$ and $M_f^g$, we conclude that for any $x,y \in [q]^n$,
\begin{align*}
\cM_{q \cdot n, s}(\cI_y, \tilde{x}) & = c - \cI_y(\tilde{x}) \mcom\\
& = c - \cI^*\left( \tilde{x}_{1,y_1},\ldots, \tilde{x}_{1,y_n} \right) \mcom\\
& = c - \cI^* \left( (g \circ x_1)_{y_1},\ldots, (g \circ x_n)_{y_n} \right) \mcom \\
& = c - \cI^* \left( g(x_1,y_1),\ldots, g(x_n,y_n) \right) = f(g^n(x,y))\mcom \\
& = M_{f}^g(x,y) \mper
\end{align*}
Therefore, the matrix $M_{f}^g$ is a sub-matrix of $\cM_{q.n,s}$ as desired.
\end{proof}

Now we are ready to wrap up the proof of our main result \prettyref{thm:maincsp}, concerning the optimality of Sherali-Adams linear programs, among all linear programs of roughly the same size.
\begin{proof}[Proof of \prettyref{thm:maincsp}]
Suppose $f(n)$-round Sherali-Adams relaxation for a CSP $\Lambda$ does not achieve a $(c,s)$-approximation on instances with $n$ variables.  This implies that there exists instances $\cI_n$ on $n$ variables such that $\opt(\cI) \leq s$ but the optimum value of $f(n)$-round Sherali-Adams linear program $\opt_{SA(f(n))}(\cI) > c$.

Set $I_n(x) \seteq c - \frac{1}{n} -  \cI_n(x)$.
The work of Chan \etal \cite{CLRS13} observes that the dual to the $f(n)$-round Sherali-Adams linear program corresponds to expressing the function $c - \cI_n(x)$ as a sum of non-negative $f(n)$-juntas.  In particular, this implies that $\degp(I_n + \frac{1}{n}) \geq f(n)$.  

Applying \prettyref{thm:nnrlift} to function $h_n$ we get that,
\[ \nnr(M_{I_n}^b) \geq 2^{\Omega(b \cdot \degp(h_n + \frac{1}{n}))}\]
for some $b = \Theta(\log n)$.  By \prettyref{lem:submatrix}, the matrix $M_{I_n}^b$ is a sub-matrix $\cM_{n\cdot q,s}$ for $q = 2^b$.  Therefore we get,
\[ \nnr(\cM_{n^H,s}) \geq n^{h \cdot \degp(I_n + \frac{1}{n})} \geq n^{h \cdot f(n)} \mcom\]
for some constants $h, H \in \N$.  Using \prettyref{lem:lplb-via-pattern}, this implies that no linear program of size $n^{h \cdot f(n)}$ can $(c-\frac{1}{n},s)$-approximate the CSP $\Lambda$ on instances with $n^{H}$ variables.  
\end{proof}

We now prove \prettyref{thm:maincsp}, \prettyref{cor:maincsp}. For this, we need the following results on the performance of the Sherali-Adams hierarchy for CSPs. Charikar et. al. \cite{CharikarMM09} showed the following lower bound for MAX-CUT.
\begin{theorem}[Sherali Adams Integrality Gaps \cite{CharikarMM09}]
\label{thm:sherali-adams-cmm}
For every $\epsilon > 0$, there is a $\gamma = \gamma(\epsilon)$ such that the $n^{\gamma}$-round Sherali-Adams relaxation for MAX-CUT does not achieve a $(1/2+\epsilon,1-\epsilon)$-approximation\footnote{The results of \cite{CharikarMM09} are actually stated in terms of integrality gaps, but their proofs actually show this stronger statement.}. 
\end{theorem}

Grigoriev \cite{Gri01} showed a lower bound for the \emph{Sum-of-Squares SDP hierarchy} (that is a strengthening of the Sherali-Adams LP hierarchy and thus the lower bounds carry over) for 3XOR; Schoenebeck \cite{Schoenebeck08} rediscovered this result and also observed that it implies a similar lower bound for the MAX-3SAT problem. Following this, \cite{BGMT12} extended this result to show a $\Omega(n)$-round lower bound for every \emph{pairwise independent} CSP; here, a CSP defined by a predicate $P:\zo^k \rightarrow \zo$ is pairwise independent if there exists a balanced pairwise independent distribution $\mu$ supported on $P^{-1}(1)$). 

\begin{theorem}[\cite{Gri01,Schoenebeck08,BGMT12}]\label{thm:sacsp}
For every $k$-ary pairwise independent predicate $P$ and $\epsilon > 0$,  there exists a constant $c = c(k, \epsilon)$ such that the $cn$-round Sherali-Adams relaxation for MAX-CSP problem on predicate $P$ achieves a $(|P^{-1}(1)|/2^k +\epsilon, 1-\epsilon)$-approximation. As a corollary, for some constants $c_1(\epsilon), c_2(\epsilon)$, the $c_1(\epsilon)$-round Sherali-Adams relaxation for MAX-3SAT does not achieve a $(7/8+\epsilon,1-\epsilon)$-approximation, and $c_2(\epsilon)$-round Sherali-Adams relaxation for MAX-3XOR does not achieve a $(1/2+\epsilon,1-\epsilon)$-approximation.
\end{theorem}

\begin{proof}[Proof of \prettyref{cor:maincsp}]
Consider the case of MAX-3SAT. Then, combining the above theorem with \prettyref{thm:maincsp} we get that any LP relaxation for MAX-3SAT of size $n^{h c_1(\epsilon) n}$ on $n^H$-variables does not achieve a $(7/8+\epsilon, 1-\epsilon)$-approximation. Let $N = n^H$. Then, this says that no LP relaxation for MAX-3SAT of size $N^{h c_1(\epsilon) N^{1/H} /H} = N^{\Omega_\epsilon(N^{1/H})}$ achieves a $(7/8+\epsilon,1-\epsilon)$-approximation. The latter condition in particular implies that such LP relaxations have integrality gap at least $1-\epsilon/(7/8+\epsilon) = 8/7 - O(\epsilon)$. This implies the claimed lower bound for MAX-3SAT. The claims for MAX-3XOR, MAX-CUT, and pairwise independent predicates follow similarly. 
\end{proof}


\begin{proof}[Proof of \prettyref{thm:approxvsexact}]
We only sketch the argument here and refer to \cite{CLRS13} where the connection between such separations and lower bounds for CSPs as above is drawn out in more detail.

The theorem essentially follows by showing corresponding separations for \emph{degrees} and using our lifting theorem. Let $f:\on^n \to \rnng$ be a function and let $\deg(f)$ be the degree of $f$ as a polynomial and for all $\delta > 0$, let 
$$\degp^\delta(f) = \min\{\degp(h): \|h-f\|_\infty \leq \delta \|f\|_\infty\}.$$

Then, for all $b$, $\rank(M_f^b) \leq \binom{n}{\deg(f)} \cdot 2^{b \cdot \deg(f)}$, $\nnr^\delta(M_f^b) \leq \binom{n}{\degp^\delta(f)} \cdot 2^{b \cdot \degp^\delta(f)}$. Thus, to show the theorem, it would suffice to find a function $f:\on^n \to \rnng$ such that $\ex[f] = \Omega(1)$, $\deg(f) = O(1)$, $\degp^\epsilon(f) = O(\log(1/\epsilon))$ and $\degp(f+\epsilon) = \Omega_\epsilon(n)$. For, if we take $M = M_f^b$ for $b \geq C\log n$ as in \prettyref{thm:nnrlift}, then $M \in \rnng^{N \times N}$ for $N = 2^{bn}$ and 
\begin{align*}
\rank(M) &\leq \binom{n}{\deg(f)} \cdot 2^{b \cdot \deg(f)} = (\log N)^{O(1)}\\
\nnr^\epsilon(M) &\leq \binom{n}{\degp^\epsilon(f)} \cdot 2^{b \cdot \degp^\epsilon(f)} = n^{O(\log(1/\epsilon)} = (\log N)^{O(\log(1/\epsilon))},\\
\nnr(M) &= \exp(\Omega(b \cdot \degp(f + O(1/n)))) = \exp(\Omega(b \cdot \degp(f+\epsilon))) = N^{\Omega_\epsilon(1)}.
\end{align*} 

To show the existence of such a function, let $\I$ be an instance of MAX-3SAT on $n$-variables such that $\opt(\I) \leq 7/8 + \epsilon$, but the $\Omega_\epsilon(n)$-round Sherali-Adams relaxation has value at least $1-\epsilon$. Such instances exist by \cite{Gri01, Schoenebeck08}. Define $f:\on^n \to \rnng$ as $f(x) = 1 - 2\epsilon - \I(x)$. (This is non-negative valued for $\epsilon < 1/24$.)

Clearly, $\deg(f) = 3$ and by the relation between Sherali-Adams relaxations and $\degp$, $\degp(f+\epsilon) = \degp(1 - \epsilon -\I(\;)) = \Omega_\epsilon(n)$. To finish the proof, it remains to show that $\degp^\epsilon(f) = O(\log(1/\epsilon))$. This follows from a similar argument used in \cite{CLRS13} for MAX-CUT. 
\end{proof}
\RMnote{Prove Theorems 1.2, Cor 1.3, 1.8 here. Basically need references.}

\section*{Acknowledgements}
We thank the Simons Institute for Theory of Computing where part of this work was done as part of the program on ``Fine-grained Complexity'' in Fall 2015. We thank Mika G\"o\"os, James Lee, David Steurer, Thomas Vidick for valuable comments.
We thank Ryan O'Donnell and Yu Zhao for pointing out bugs in a previous version of this paper.
 \addreferencesection 
\bibliographystyle{amsalpha}
\bibliography{BIB/mr,BIB/dblp,BIB/scholar,BIB/zblatt,BIB/ads,BIB/custom}

\appendix
\end{document}